\documentclass{article}
\usepackage[latin9]{inputenc}

\usepackage{upgreek}
\usepackage{amsmath}
\usepackage{amssymb}

\makeatletter
\usepackage{amstext}
\usepackage{amsfonts}
\usepackage{amsthm}
\usepackage{bbold}
\usepackage{mathtools}
\usepackage{scrextend}
\usepackage{array}
\usepackage{multicol}
\usepackage{mathrsfs}
\usepackage{color}
\usepackage{soul}

\newcommand{\suppress}[1]{}

\newtheorem{theorem}{Theorem}[section]    % Specify Theorem
\newtheorem{definition}{Definition}[section] % Specify Definition
\newtheorem{corollary}[theorem]{Corollary}    % Specify Corollary
\newtheorem{proposition}[theorem]{Proposition}    % Specify Lemma
\newtheorem{lemma}[theorem]{Lemma}    % Specify Lemma
\renewcommand{\qed}{\hfill{$\rule{6pt}{6pt}$}} %Box at end of proof
\renewenvironment{proof}{\noindent{\bf Proof:}}{\qed}

\newenvironment{proofof}[1]{\noindent{\bf Proof of #1:}}{\qed}
\newenvironment{remark}{\noindent{\bf Remark:}}

\numberwithin{equation}{section}

\newcommand{\complex}{{\mathbb C}}
\newcommand{\reals}{{\mathbb R}}

\newcommand{\ket}[1]{| #1 \rangle}
\newcommand{\bra}[1]{ \langle #1 |}
\newcommand{\ketbra}[2]{| #1 \rangle\!\langle #2 |}
\newcommand{\braket}[2]{\langle #1 | #2 \rangle }
\newcommand{\innerproduct}[2]{\langle #1 , #2 \rangle }

\newcommand{\size}[1]{\left| #1 \right|}
\newcommand{\set}[1]{\left\{ #1 \right\}}

\newcommand{\ceil}[1]{{\left\lceil #1 \right\rceil}}
\newcommand{\trace}{{\mathrm{Tr}}}
\newcommand{\support}{{\mathrm{supp}}}

\newcommand{\density}[1]{\ketbra{#1}{#1}}

\newcommand{\transpose}{{\mathsf T}}
\newcommand{\Span}{{\mathrm{span}}}

\newcommand{\abs}[1]{\left| #1 \right|}

\newcommand{\id}{{\mathbb 1}}

\newcommand{\tensor}{\otimes}

\newcommand{\union}{\cup}
\newcommand{\intersect}{\cap}

\newcommand{\comment}[1]{}

\newcommand{\etal}{\emph{et al.\/}}

\newcommand{\cT}{{\mathcal T}}

\newcommand{\rI}{{\mathrm I}}
\newcommand{\rS}{{\mathrm S}}
\newcommand{\rH}{{\mathrm H}}
\newcommand{\rL}{{\mathrm L}}
\newcommand{\rF}{{\mathrm F}}

\newcommand{\rD}{{\mathrm D}_{\mathrm{obs}}}
\newcommand{\Imax}{{\mathrm I}_{\max}}
\newcommand{\Dmax}{{\mathrm D}_{\max}}
\newcommand{\Dh}{{\mathrm D}_{\mathrm{h}}}
\newcommand{\rP}{{\mathrm P}}
\newcommand{\rd}{{\mathrm d}}

\newcommand{\sD}{{\mathsf D}}
\newcommand{\sPos}{{\mathsf{Pos}}}
\newcommand{\sQ}{{\mathsf Q}}
\newcommand{\sq}{{\mathsf q}}
\newcommand{\sT}{{\mathsf T}}
\newcommand{\sB}{{\mathsf B}}
\newcommand{\sL}{{\mathsf L}}
\newcommand{\sU}{{\mathsf U}}

\newcommand{\cK}{{\mathcal K}}
\newcommand{\cH}{{\mathcal H}}
\newcommand{\cM}{{\mathcal M}}

\newcommand{\scE}{{\mathscr E}}

\newcommand{\RSP}{{\mathrm{RSP}}}

\makeatother

\begin{document}

\title{ \textbf{ Communication Complexity of One-Shot Remote State
Preparation \footnote{Much of the work in this article was reported in
S.B.'s Master's thesis \cite{bab_hadiashar_communication_2014}. }} \\
}

\date{}

\author{
Shima Bab Hadiashar~\thanks{Department of Combinatorics and Optimization, and Institute for Quantum Computing, University
of Waterloo, 200 University Ave.\ W., Waterloo, ON,
N2L~3G1, Canada. Email: \texttt{sbabhadi@uwaterloo.ca}~.
Research supported in part by NSERC Canada.} \\
U.\ Waterloo
\and
Ashwin Nayak~\thanks{Department of Combinatorics and Optimization,
and Institute for Quantum Computing,
University of Waterloo, 200 University Ave.\ W., Waterloo, ON,
N2L~3G1, Canada. Email: \texttt{ashwin.nayak@uwaterloo.ca}~.
Research supported in part by NSERC Canada.} \\
U.\ Waterloo
\and
Renato Renner~\thanks{Institute for Theoretical Physics, ETH Zurich,
Wolfgang-Pauli-Str.\ 27, 8093 Zurich, Switzerland.
Email: \texttt{renner@itp.phys.ethz.ch}~. } \\
ETH Zurich
}

\maketitle

\begin{abstract}
Quantum teleportation uses prior shared entanglement and classical
communication to send an unknown quantum state from one party to
another. Remote state preparation (RSP) is a similar distributed task
in which the sender knows the entire classical description of the state
to be sent. (This may also be viewed as the task of
\emph{non-oblivious\/} compression
of a single sample from an ensemble of quantum states.) We study the communication complexity of
\emph{approximate\/} remote state preparation, in which the goal is to
prepare an approximation of the desired quantum state.

Jain [Quant. Inf. \& Comp., 2006] showed that
the worst-case communication complexity of approximate RSP can be
bounded from above in terms of the \emph{maximum possible information\/} in an
encoding.  He also showed that this
quantity is a lower bound for communication complexity of
(\emph{exact\/}) remote state preparation. 
In this work, we tightly characterize the worst-case and average-case
communication complexity of remote state preparation in terms of
non-asymptotic information-theoretic quantities.

We also  show that the
average-case communication complexity of RSP can be much smaller
than the worst-case one. In the process, we show that~$n$ bits 
cannot be communicated with less than~$n$ transmitted bits in
LOCC protocols. This strengthens a result due to Nayak and
Salzman [J. ACM, 2006] and may be of independent interest.
\end{abstract}

\section{Introduction}

Quantum teleportation~\cite{bennett_teleporting_1993} is an archetypical
protocol in information processing that is
impossible in the absence of quantum resources like shared entanglement.
Through quantum teleportation, one party is able to communicate an
\emph{arbitrary\/} qubit state to another party using only two classical bits of
communication and a previously shared maximally entangled pair of qubits.
The two classical bits of communication and
a maximally entangled pair of qubits are both necessary and sufficient
for the task. This is a remarkable phenomenon, as the entire classical
description of the state being communicated is potentially infinite in
length.

In Ref.~\cite{lo_classical-communication_2000}, Lo introduced a similar
distributed task in which the sender (called Alice in the literature)
knows a classical description of the quantum state. This task is called
\emph{remote state preparation} (RSP). In particular, remote state
preparation is a task involving two
parties, Alice and Bob, who share qubits in an entangled state. Alice is given
the description of a state,~$Q(x)$, chosen from a subset of quantum states~$\{Q(1),\ldots,Q(n)\}$, and
their goal is to prepare that quantum state on Bob's side using only
local operations and classical communication (LOCC).
This may also be viewed as the task of compression (which is non-oblivious 
at the sender's end),
of a single sample from an ensemble of quantum states with
entanglement-assisted classical communication.

We say an RSP protocol is \textit{oblivious to Bob} if he can get no more information
about the prepared state than what is contained in a single copy of the
state~\cite{leung_oblivious_2003}. A relaxed version of RSP is
\textit{approximate remote state preparation} (ARSP) in which we wish to
prepare an approximation~$\sigma_x$ of the specified quantum state~$Q(x)$. We define the error of a protocol for approximate remote state preparation in terms of the fidelity between~$Q(x)$ and~$\sigma_x$. We say a protocol has \textit{worst-case error} at most~$\epsilon$, if for every~$x\in\{1,\ldots,n\}$,~$\rF(Q(x),\sigma_x)\geq \sqrt{1-\epsilon^2}$. Similarly, a protocol has \textit{average-case error} at most~$\epsilon$ with respect to a probability distribution~$p$, if~$\sum_{x=1}^n p_x \rF(Q(x),\sigma_x)\geq \sqrt{1-\epsilon^2}$.

Lo~\cite{lo_classical-communication_2000} gave several examples
of ensembles which can be remotely prepared using a one-way communication protocol
with classical communication cost less than that in quantum teleportation.
However, he conjectured that to prepare arbitrary pure $n$-qubit states remotely,
Alice has to necessarily send the same number of classical bits as in
quantum teleportation i.e., $2n$ classical bits. The task has been
studied extensively since then, largely in the asymptotic setting.

Bennett~$\etal$~\cite{bennett_remote_2001} showed that in the presence
of a large amount of shared entanglement, Alice can prepare general
quantum states on Bob's side with the asymptotic classical communication
rate of one bit per qubit. This amount of classical communication from
Alice to Bob is also necessary by
causality~\cite{lo_classical-communication_2000}. They also showed that
unlike for quantum teleportation,
there is a trade-off between the communication cost and the amount
of entanglement in remote state preparation. In particular, they proved
that at the cost of using more entanglement, the communication cost of
preparing a one-qubit state  ranges from one bit in the high
entanglement limit to an infinite number of bits in the case of no
previously shared entanglement. In addition,
they suggested that the Lo conjecture  is true in a more restricted setting, such as when the protocol is \textit{faithful}
and oblivious to Bob~\cite{bennett_remote_2001}.
(A protocol is said to be \textit{faithful} if it is exact and
deterministic.)

Devetak and Berger~\cite{devetak_low-entanglement_2001} found an
analytic expression for the trade-off curve between the shared
entanglement and classical communication of \textit{teleportation based
RSP protocols} in the \textit{low-entanglement} region (less than~$1$
singlet state per qubit). They conjectured that teleportation based protocols are optimal among all low-entanglement protocols. 
Later, Leung and Shor~\cite{leung_oblivious_2003} proved the Lo
conjecture for a special case. They proved that if a one-way RSP
protocol for a $\emph{generic ensemble}$ of pure states is faithful and
oblivious to Bob, then it necessarily uses at least as much classical
communication as in teleportation. (A $\emph{generic ensemble}$ is an
ensemble of states whose density matrices span the operators in the
input Hilbert space.)
Hayashi, Hashimoto and Horibe~\cite{hayashi_remote_2003} showed that in
order to remotely prepare one qubit in an arbitrary state using a one-way
faithful, but not necessarily oblivious protocol, Alice requires 
two classical bits of communication as in teleportation. 
 
Berry and Sanders~\cite{berry_optimal_2003} studied ARSP, the
approximation variant of RSP, of an
ensemble~$\scE$ of mixed states (which might be entangled with some
other system on Alice's part) such that their entanglement with other
systems does not change significantly. They showed that approximate remote state preparation with arbitrary small average-case error~$\epsilon$ can be done asymptotically using communication per prepared state arbitrarily close to the Holevo information~$\chi(\scE)$ of the ensemble.  
(See Section~\ref{sec:Asymptotic Information Theory} for a definition of
Holevo information.)
Later Bennett, Hayden, Leung, Shor, and
Winter~\cite{bennett_remote_2005} proved that approximate remote state
preparation with small worst-case error~$\epsilon$ requires an asymptotic rate of one bit of classical communication per qubit from Alice to Bob. They also showed that this amount of classical communication is sufficient. Moreover, they derived the exact trade-off curve between shared entangled bits and classical communication bits for an arbitrary ensemble of candidate states.

Jain~\cite{jain_communication_2006} studied remote state preparation in
the one-shot scenario. He considered the total communication cost when
given access to an arbitrary amount of entanglement.
He showed that the communication cost required for exact remote state
preparation is at least~$\sT(Q)/2$ and  ARSP with worst-case error at
most~$\epsilon$ can be accomplished with communication at most~$\frac{8}{\left(1-\sqrt{1-\epsilon^{2}}\right)^2}(4\sT(Q)+7)$, where~$\sT(Q)$
denotes the \emph{maximum possible information\/} in an encoding~$Q$. 
(A  precise definition can be found in Section~\ref{sec:Asymptotic
Information Theory}.)

These abovementioned results on remote state preparation are summarized in Table~\ref{table1} .
\begin{table}

\begin{tabular}{|>{\centering}m{3.2cm}|>{\centering}m{3.8cm}|>{\centering}p{3cm}|>{\centering}p{4cm}|}
\hline 
Protocol Type & Conditions & Entanglement & Classical Communication\tabularnewline
\hline 
\hline 
Faithful RSP~\cite{bennett_remote_2001} & an arbitrary state,

one-way communication,

in asymptotics & high entanglement & = 1 classical bit per qubit\tabularnewline
\hline 
Faithful RSP~\cite{hayashi_remote_2003} & one pure qubit in a general state,

one-way communication & = 1 ebit(singlet) per qubit & = 2 classical bit \tabularnewline
\hline 
Faithful and oblivious RSP~\cite{leung_oblivious_2003} & a generic ensemble of pure states, 

one-way communication & = 1 ebit(singlet) per qubit & = 2 classical bit per qubit\tabularnewline
\hline 
ARSP with small average-case error~\cite{berry_optimal_2003} & an ensemble $\scE$ of mixed states preserving their entanglement,

one-way communication,

in asymptotics & no limit & $\approx\chi(\scE)$ classical bits per prepared state$ $\tabularnewline
\hline 
ARSP with small worst-case error~\cite{bennett_remote_2005} & an arbitrary pure state, two-way communication,

in asymptotics & = 1 ebit(singlet) per qubit & = 1 classical bit per qubit from Alice to Bob\tabularnewline
\hline 
Exact RSP~\cite{jain_communication_2006} & an arbitrary state, two-way communication,

in one-shot scenario & no limit & $\geq\sT(Q)/2$\tabularnewline
\hline 
ARSP with worst-case error $\epsilon$~\cite{jain_communication_2006} & an arbitrary state, one-way communication,

in one-shot scenario & no limit & $\leq\frac{8}{\left(1-\sqrt{1-\epsilon^{2}}\right)^{2}}(4\sT(Q)+7)$\tabularnewline
\hline 
\end{tabular}
\caption{A summary of previous works on communication cost of Remote State Preparation}\label{table1}
\end{table}

\subsection{Our results}

Intuitively, relaxing the remote state preparation problem so that Bob
produces some approximation to the ideal state should lower
the communication complexity of the task. This suggests that
the bounds provided by Jain~\cite{jain_communication_2006} are not
tight.

In this work, we characterize the communication complexity of remote
state preparation in two different cases. First, we consider ARSP with
average-case error at most~$\epsilon$, and bound its communication
complexity by the \emph{smooth max-information\/}
Bob has about Alice's input.
(See Section~\ref{sec:Asymptotic Information
Theory} for a precise definition of this quantity.)
Then we consider ARSP with worst-case error at most~$\epsilon$,
and give lower and upper bounds for its communication complexity in
terms of \emph{smooth max-relative entropy\/} and show that these bounds may be 
arbitrarily tighter than that in Ref.~\cite{jain_communication_2006}.

Our main results about the remote state preparation problem are summarized
below, using notions introduced in Section~\ref{sec-preliminaries}.
Recall that a protocol has worst-case error at most~$\epsilon$, if for every~$x\in\{1,\ldots,n\}$, $\rF(Q(x),\sigma_x)\geq \sqrt{1-\epsilon^2}$, and a protocol has average-case error at most~$\epsilon$ with respect to a probability distribution~$p$, if~$\sum_{x=1}^n p_x \rF(Q(x),\sigma_x)\geq \sqrt{1-\epsilon^2}$. We denote the average-case communication complexity of ARSP by~$\sQ^*_p(\RSP(S,Q),\epsilon)$, and the worst-case communication complexity of ARSP by~$\sQ^*(\RSP(S,Q),\epsilon)$.

\begin{theorem}\label{thm-main theorem} 
For any finite set~$S$, and
set of quantum states~$\{Q(x) : x \in S \}$, let~$p$
be a probability distribution over~$S$ and~$\rho_{AB}(p) \in \sD(\cH' \tensor \cH)$ be the bipartite quantum state~$\rho_{AB}(p)=\sum_{x\in S} p_{x} \density{x}_{A} \tensor Q(x)_{B}$.
Then
\begin{enumerate}
\item For any fixed~$\epsilon \in (0,1]$, we have

\[
 \Imax^{\epsilon}(A:B)_{\rho(p)}\quad\leq\quad\sQ_{p}^{*}(\RSP(S,Q),\epsilon)\quad\leq\quad \Imax^{\frac{\epsilon}{2\sqrt{2}}}(A:B)_{\rho(p)}+f(\epsilon)\enspace,
\]
where~$f(\epsilon)\in \Theta(\log \log \frac{1}{\epsilon})$ is a
function of~$\epsilon$, and~$\Imax^\epsilon(A:B)$ denotes the smooth max-information part B has about part A.

\item For any fixed~$\epsilon \in (0,1]$ and for any $0<\delta< 1 -
\epsilon^2$, we have
\begin{align*}
\min_{\sigma\in\sD(\cH)}
\max_{x\in S}\ 
\Dmax^{\sqrt{2(\epsilon^{2}+\delta)}}(Q(x)\|\sigma)+g_{1}(\epsilon,\delta)\quad\leq & \quad\sQ^{*}(\RSP(S,Q),\epsilon)\\
\quad\leq & \quad\min_{\sigma\in\sD(\cH)}\ 
\max_{x\in S}\ 
\Dmax^{\frac{\epsilon}{\sqrt{1+\epsilon^{2}}}}(Q(x)\|\sigma)+g_{2}(\epsilon)\enspace,
\end{align*}
where~$g_1, g_2$ are functions such that~$g_{1}(\epsilon,\delta)\in\Theta\left( \log\frac{\delta^3}{\epsilon^2+\delta} \right) $, $g_{2}(\epsilon)\in \Theta(\log \log \frac{1}{\epsilon})$, and~$\Dmax^\epsilon(Q(x)\|\sigma)$ denotes the smooth max-relative entropy of~$Q(x)$ with respect to~$\sigma$.
 
\end{enumerate}
\end{theorem}
It is relatively straightforward to show that the one-shot information
expressions appearing in the above theorem are continuous in~$\epsilon$.
This indicates the tightness of the characterization.
In fact, a bound on the difference between lower and upper bounds in the above theorem, in
terms of the ensemble, may be inferred from the continuity property.

We remark that the quantity appearing in the second part of the theorem
is similar to the notion of \emph{information radius\/}. It may be possible
to relate the quantity to smooth max-information with respect to 
a distribution over~$S$ 
using ideas from Ref.~\cite[Lemma 3]{DW17} (which extends
Ref.~\cite[Lemma~14]{WWY14}), and the connection between max-relative
entropy and the \emph{sandwiched R{\'e}nyi relative entropy\/}.
Finally, earlier works have considered remote state preparation 
of states drawn from infinite sets of states.
We discuss how the bounds in Theorem~\ref{thm-main theorem}
may be applied to that case in Appendix~\ref{sec-RSP-infinite-set}.

The communication cost of ARSP may decrease dramatically when more error is allowed,and if we consider average-case
error instead of worst-case error. In particular, we show that for
every~$\epsilon\in[0,\frac{1}{\sqrt{2}})$, there exists a set of~$n$
quantum states for which there is a~$\log n$ gap between the worst-case error and average-case error remote preparation of that set. In
addition, for a special set of quantum states, we derive a gap between
the worst-case error and average-case communication complexity in
terms of~$\epsilon$. This confirms our intuition that the more skewed the probability distribution is, the bigger the gap between worst-case 
and average-case error variants may be. 

In the process of establishing the first gap described above,
we strengthen a result due to Nayak and Salzman~\cite{nayak_limits_2006};
we prove a bound on the communication required by any LOCC protocol 
for transmitting a uniformly random~$n$ bit string with some probability~$p$. 
This bound is optimal, and may be of independent interest.
\begin{theorem}
\label{thm-convey info by LOCC}
Let~$Y$ be the output of Bob in any two-way LOCC protocol in which Alice
receives a uniformly distributed~$n$-bit input~$X$ (that is not known to
Bob, and is independent of their joint quantum state). Let~$m_A$ be the 
total number of bits Alice sends to Bob and~$p \coloneqq \Pr[Y=X]$ be 
the probability that Bob obtains the output~$X$. Then
\[
m_{A}\quad\geq\quad n+\log p \enspace.
\]
\end{theorem}

Worst-case protocols for ARSP capture precisely the task of compression
in one-way communication complexity.
Average-case protocols for ARSP are relevant in the distributional
setting in communication complexity, and in asymptotic information
theory.  The results in this paper thus
supercede those due to Jain, Radhakrishnan, and Sen~\cite{jain_prior_2005}
(and due to Touchette~\cite{Touchette15} for the same setting). We also
show how a characterization
due to Berry and Sanders~\cite{berry_optimal_2003} may
be reproduced from ours, via a quantum asymptotic equipartition
property (cf.\ Theorem~\ref{thm-QAEP-Imax}). Thus, we believe the
results presented here have wider ramifications.

\subsection{Organization}

The organization of this paper is as follows. In
Section~\ref{sec-preliminaries}, we review some concepts, fix notation,
and the terminology used in the paper. Then we define remote state preparation, and explain an efficient protocol for this problem introduced in Ref.~\cite{jain_prior_2005}. In Section~\ref{sec:Average-case-error-communication}
and Section~\ref{sec:Worst-case-error-communication}, we give bounds on
average-case error and worst-case communication complexity of
ARSP, respectively. We make some observations, including a comparison
with previously known results in Section~\ref{sec:Some-Observations}.
We analyze LOCC protocols for communicating a uniformly random~$n$ bit
string in Section~\ref{sec-locc-bits}.
The paper ends with a summary of our results and an outlook in
Section~\ref{part:Conclusions-and-Outlook}. In the Appendix,
we present the proofs of some properties of information-theoretic quantities, and
discuss remote state preparation of states drawn from an infinite set.

\section*{Acknowledgments}

We are grateful to Matthias Christandl for discussions which led to the
research reported in this article.
S.B.\ thanks Marco Tomamichel for his help with the proof of 
Theorem~\ref{thm-worst-case lb} during her internship at CQT, Singapore.
We also thank the reviewers and the Associate Editor, Mark Wilde, for
their comments and suggestions.

\section{Preliminaries\label{sec-preliminaries}}

In this section, we review some notions in quantum computing and quantum information theory, such as  LOCC protocols, quantum communication complexity, asymptotic and non-asymptotic quantum information theory, as well as some mathematical tools like the minimax theorem. We also define remote state preparation formally and describe a non-trivial protocol for this problem. We refer the reader to the books by
Nielsen and Chuang~\cite{nielsen_quantum_2000} and Watrous~\cite{watrous-2015} for basic notions and results in quantum information, and largely only describe the potentially
non-standard notation and terminology we use.

\subsection{Some basic notions}

We denote Hilbert spaces either by capital script letters like~$\cH$ and $\cK$, or as~$\complex^m$ where~$m$ is the dimension of the Hilbert space. We
concern ourselves only with finite dimensional Hilbert spaces in this
article.
We denote the set of all linear operators from~$\cH$ to~$\cK$ by~$\sL(\cH,\cK)$.
We abbreviate~$\sL(\cH,\cH)$ as~$\sL(\cH)$. 
 We denote the set of all positive semidefinite operators in~$\cH$ by~$\sPos(\cH)$. An operator~$A$ is called \emph{sub-normal} if it is positive semidefinite and has trace at most 1. 
(The term ``subnormalized'' is also often used for such operators.)

We denote the identity operator on a Hilbert space by~$\id$ and the set of all unitary
operators on space~$\cH$ by~$\sU(\cH)$.

We call a physical quantum system with a finite number of degrees of
freedom a \emph{register}. Every register is associated with a Hilbert
space. We denote registers by capital letters, e.g.,~$X$,~$Y$ and~$Z$.
We use the notation~$\abs{X}$ to denote the dimension of the Hilbert
space associated with register~$X$. The state of a register~$X$ is modelled as a \emph{density operator}, i.e., a positive semidefinite operator with trace one, and is called a \emph{quantum state}. We denote  density operators by lower case Greek letters (e.g.,~$\rho$,
$\sigma$, \ldots), and the set of all density operators over a Hilbert space~$\cH$ by~$\sD(\cH)$. We may also denote a state by~$\rho_X$ to indicate its register~$X$. A bipartite register~$XY$ with Hilbert space~$\cH\tensor\cK$
is called a \textit{classical-quantum} register in the context of an
information processing task, if it only assumes states of the form~$\sum_{i}p_{i} \density{e_{i}} \tensor \rho_i$ where~$\{\ket{e_i}\}$ is the standard basis of~$\cH$ and~$p$ is a probability distribution over the basis.
In that case we say that the states are classical on~$X$. 
For any~$\omega \in \sPos(\cH)$ with spectral decomposition~$\sum_i
\lambda_i \density{\psi_i}$, we let~$\sqrt{\omega}=\sum_i \sqrt{\lambda_i}\density{\psi_i}$.

We denote the partial trace over Hilbert space~$\cK$ of a quantum state~$\rho_{AB}
\in \sD(\cH \tensor \cK)$ by either~$\trace_{\cK}(\rho_{AB})$  or~$\trace_B(\rho_{AB})$. We say that~$\rho_{AB}\in\sD(\cH\otimes\cK)$
is an \emph{extension }of~$\rho_{A}\in \sD(\cH)$ if~$\trace_{\cK}(\rho_{AB})=\rho_{A}$.

We call completely positive and trace preserving linear maps~$\sL(\cH) \rightarrow \sL(\cK)$ \emph{quantum channels\/}.
Quantum measurements are quantum channels with Kraus operators~$\{\sqrt{E_a}\otimes\ket{a}: a\in \Gamma\}$, where~$\Gamma$ is the set
of outcomes of the measurement and~$E_a$ is a positive semidefinite
operator associated with the outcome~$a\in\Gamma$ such
that~$\sum_{a \in \Gamma} E_a = \id$. We refer to the operators~$E_a$ as
\emph{measurement operators\/}.

The \emph{fidelity\/} $\rF(\rho,\sigma)$ between two quantum states~$\rho$ and~$\sigma$, is defined as
\[\rF(\rho,\sigma) \quad \coloneqq \quad \trace\sqrt{\sqrt{\rho}\ \sigma\sqrt{\rho}}\enspace.
\]
In the literature, fidelity is sometimes defined as the square of the
above quantity. Fidelity may be extended to sub-normal states~$\rho,\sigma$ as follows:
\[
\rF(\rho,\sigma)\quad\coloneqq\quad\trace\sqrt{\sqrt{\rho}\ \sigma\sqrt{\rho}}+\sqrt{\left(1-\trace(\rho)\right)\left(1-\trace(\sigma)\right)}\enspace.
\]
The fidelity function is monotone under the application of quantum channels, and is jointly concave over the set of quantum states.
Other useful properties of fidelity are stated in the following propositions.

\begin{proposition}\label{eq:fidelity-trace}
For any quantum state~$\rho$ and sub-normal state~$\sigma$, it holds that 
\[
\rF(\rho,\sigma)^2\quad \leq\quad \trace(\sigma)\enspace.
\]
\end{proposition}

\begin{proposition}
\label{eq:triangle-ineq-fidelity}
Let~$\rho, \sigma\in\sD(\cH)$ be two quantum states. Then
\[
1+\rF(\rho,\sigma)\quad=\quad\max\  \{\rF(\rho,\xi)^2+\rF(\sigma,\xi)^2 : \xi\in\sD(\cH)\}\enspace.
\]
\end{proposition}
For a proof of the above property, see Ref.~\cite[Lemma~3.3]{nayak_bit-commitment_2003}.

We use the \textit{purified distance} (see
Ref.~\cite{tomamichel_duality_2010}) as a metric for sub-normal states. 
This is an extension of the metrics developed in Refs.~\cite{rastegin2002relative,rastegin2003lower, Gilchrist2005Distance, rastegin2006sine}.
Suppose that~$\rho$ and~$\sigma$ are two sub-normal states.
Then the purified distance of~$\rho$ and~$\sigma$ is defined as
\[
\rP(\rho,\sigma)\quad\coloneqq\quad\sqrt{1-\rF(\rho,\sigma)^{2}}\enspace.
\]
There are other metrics over sub-normal states, such as the trace distance.
However, we choose purified distance since it turns out to be more convenient to use in non-asymptotic quantum information theory. 

Let~$\rho\in\sD(\cH)$ be a quantum state and~$\epsilon\in[0,1)$.
Then, we define
\[
\sB^{\epsilon}(\rho) \quad \coloneqq \quad \{\tilde{\rho}\in\sPos(\cH):\rP(\rho,\tilde{\rho})\leq\epsilon,\trace{\:\tilde{\rho}}\leq1\}
\]
as the ball of sub-normal states that are within purified distance~$\epsilon$ of~$\rho$. We say that~$\sigma$ is \textit{$\epsilon$-close\/} to $\rho$, or equivalently,~$\sigma$ is an\textit{~$\epsilon$-approximation\/} of~$\rho$, if~$\sigma
\in \sB^{\epsilon}(\rho)$. The following property of purified distance
states that any state~$\rho'_A$ that is~$\epsilon$-close to~$\rho_A$ may
be extended to a state~$\rho'_{AB}$ that is~$\epsilon$-close to any
given extension~$\rho_{AB}$ of~$\rho_A$.
\begin{proposition} \label{cor-uhlmann} Let~$\rho_{A}\in\sD(\cH_A)$ be
a quantum state in the Hilbert space~$\cH_A$ and~$\rho_{AB}\in\sD(\cH_A\tensor\cH_B)$
be an extension of~$\rho_{A}$ over the Hilbert space~$\cH_A\tensor\cH_B$,
i.e.~$\rho_{A}=\trace_{B}(\rho_{AB})$. Let~$\rho'_{A}\in\sB^{\epsilon}(\rho_{A})$
be an~$\epsilon$-approximation of~$\rho_{A}$. Then there exists~$\rho'_{AB}\in\sB^{\epsilon}(\rho_{AB})$
such that~$\rho'_{A}=\trace_{B}(\rho'_{AB})$.
\end{proposition}

\begin{proof}
 Let~$\ket{v}\in\sD(\cH_{A'}\tensor\cH_{B'}\tensor\cH_{A}\tensor\cH_{B})$ be a purification
of~$\rho_{AB}$ and therefore also of~$\rho_{A}$, and~$\ket{v'}\in\sD(\cH_{A'}\tensor\cH_{B'}\tensor\cH_{A}\tensor\cH_{B})$ be a purification of~$\rho'_{A}$, such that~$\rF(\rho_{A},\rho'_{A})=\size{\braket{v}{v'}}$. Such~$\ket{v}$ and~$\ket{v'}$
exist by the Uhlmann theorem. Define~$\rho'_{AB}=\trace_{A'B'}(\density{v'})$.
By definition, we have~$\rF(\rho_{A},\rho'_{A})=\rF(\rho_{AB},\rho'_{AB})$. Therefore~$\rho'_{AB}\in\sB^{\epsilon}(\rho_{AB})$.
\end{proof}

The above property is in fact an extension of the Uhlmann theorem for purified distance.

\subsection{LOCC protocols}\label{sec-LOCC protocols}

The notion of \emph{LOCC}, short for \emph{local operations and
classical communication}, plays an important role in quantum
information, especially in the study of properties of entanglement (see, e.g., Ref.~\cite{bennett1996Mixed}). This
notion has been described formally in terms of \textit{quantum instruments} in Ref.~\cite{chitambar_everything_2012}.
In this article, we only study two-party LOCC protocols, in which one
party receives a classical input, and the other party produces a quantum
output. We describe these protocols informally below.

Suppose we have two parties, Alice
and Bob, who communicate with each other using only classical bits, share parts of a possibly entangled quantum state, and are allowed to perform any local quantum channels on their registers.
We call the registers (or qubits) accessible by only one of the parties
\emph{private\/} registers (or qubits). Alice is given a classical
input; Bob does not receive any input.
Let~$A$ be the register which holds Alice's
input, $Y_{0} \coloneqq P_{0}V_{0}$ and~$Z_{0} \coloneqq Q_{0}W_{0}$ be
Alice's and Bob's initial classical-quantum private registers,
respectively. Registers~$P_{i}$ and~$Q_{i}$ are classical registers with Alice and Bob,
respectively, after the~$i$th message. These registers hold the message transcript thus
far.  Initially,~$P_0,Q_0$ are both empty.
Registers~$V_{0}$ and~$W_{0}$ are initialized to a quantum state independent of
the inputs.  Note that the state in~$V_{0} W_{0}$ might be
entangled across the registers. If there are~$k$ messages, $P_{k+1}$
and~$V_{k+1}$ denote Alice's final classical and quantum registers,
respectively, and~$Q_{k+1} W_{k+1}$ denote Bob's, potentially after a local
operation.
Register~$A$ remains unchanged throughout the protocol.
Bob produces the output, which is a sub-register~$B$ of~$Q_{k+1} W_{k+1}$.

A \emph{one-way} LOCC protocol is an LOCC protocol in which the
communication consists of one message from Alice to Bob.
The three steps of the protocol are:
\begin{itemize}
\item[1)] Alice measures her register~$V_{0}$, obtains the outcome in register~$P_{1}$
(and a residual state in~$V_1$). The measurement is controlled by her input
in~$A$.

\item[2)] Alice sends a copy of her measurement outcome
to Bob, in classical register~$M$. Bob sets~$Q_{1} = M$.

\item[3)] Bob measures his register~$W_{1}$ (which is the same as~$W_0$), controlled by the register~$Q_{1}$.
The outcome and residual state are stored in classical-quantum
registers~$Q_2 W_2$, where~$Q_2$ includes~$Q_1$.  The output
of the protocol is a designated sub-register~$B$ of his registers~$Q_2 W_2$
\end{itemize}

A \emph{two-way\/} LOCC protocol is a protocol with communication in
both directions, from Alice to Bob and Bob to Alice. It has several rounds of communication in which the two parties alternately do a local measurement and send a message. Either party may start or end the protocol. Suppose in round~$i$, it is Alice's turn. Then
\begin{itemize}
\item First, Alice measures
her quantum register in that round,~$V_{i-1}$,
controlled by her input~$A$ and her classical register~$P_{i-1}$.
She copies the outcome~$M_i$ in a fresh register~$N_i$. The
register~$P_i \coloneqq  P_{i-1} N_i$. 

\item Alice then sends~$M_i$ to Bob using~$m_{i}$
classical bits, and Bob includes the received message~$M_{i}$ in his
transcript register: $Q_{i} \coloneqq  Q_{i-1} M_i$.
\end{itemize}

Bob's actions are similar in a round in which it is his turn (except
that he does not have any input), using
registers~$Q_i W_i$. At the end
of a protocol with~$k$ rounds of communication, Bob makes a measurement 
on the quantum register~$W_{k}$ controlled by~$Q_{k}$, and he includes the
outcome~$M_{k+1}$ of the measurement in the register~$Q_{k+1}$. A
pre-designated sub-register~$B$ of~$Q_{k+1} W_{k+1}$ is the output 
of the protocol. 

\subsection{Quantum communication complexity }

Quantum communication complexity was introduced by
Yao~\cite{yao_quantum_1993}, and has been studied extensively since.
Here we describe it in the context of LOCC protocols.

Let~$X,Y$ be two finite sets,~$Z$ be a set (not necessarily finite),
and~$f\subseteq X \times Y \times Z$ be a relation such that for
every~$(x,y)\in X\times Y$, there exists some~$z\in Z$ such
that~$(x,y,z)\in f$. The sets~$X,Y,Z$ might be sets of quantum states.
For example, in remote state preparation~$Z$ is the set of quantum
states over some space. 
In an LOCC protocol, Alice and Bob get as their inputs~$x\in X$ and~$y\in Y$, respectively, and their goal is to output an element~$z\in Z$ such that~$(x,y,z)\in f$. In the protocols we consider, one party may not get any input, e.g.,~$Y$ may be empty. Also, in general the output of the protocol is probabilistic. If~$W_{x,y}$ is the random output that the protocol produces on inputs~$(x,y)$, we define the error of
the protocol as
\[
\delta \quad \coloneqq  \quad \max_{x \in X, y \in Y} \Pr((x,y,W_{x,y}) \not\in f)
\enspace.
\]
We then say the protocol \textit{computes~$f$ with error~$\delta$\/}.
\begin{definition}
The \textit{entanglement-assisted communication complexity} of~$f$ with
error~$\delta$ is defined as the minimum number of bits exchanged in an LOCC protocol computing~$f$ with error~$\delta$.
\end{definition}

Now consider a relation~$f\subseteq X\times Y\times Z$, with~$Z =
\sD(\cH)$, the set of quantum states over~$\cH$. In this context we may allow a protocol to produce an approximation to the desired quantum state. Suppose the output quantum state that an LOCC 
protocol for~$f$ produces on inputs~$(x,y)$ is denoted by~$w_{xy}$.
Let~$p$ be a probability distribution over~$X\times Y$. We say a
protocol computes an approximation of~$f$ with average-case error
at most~$\epsilon$ if there are quantum states~$\set{z_{xy} : x \in X, y \in
Y, (x,y,z_{xy}) \in f}$ such that
\[
\sum_{x \in X,y \in Y} p_{xy} \; \rF(w_{xy},z_{xy}) \quad \geq \quad
\sqrt{1 - \epsilon^2} \enspace.
\]
The above condition may equivalently be written as $\rP(\zeta, \omega) 
\le \epsilon$, where~$\zeta \coloneqq \sum_{x,y} p_{xy} \density{xy} 
\tensor z_{xy}$ is an ideal input-output state, and~$\omega \coloneqq 
\sum_{x,y} p_{xy} \density{xy} \tensor w_{xy}$ is the
actual input-output state of the protocol.

\begin{definition}
The \emph{average-case communication
complexity\/} of~$f$ is defined as the minimum number of bits exchanged in
an LOCC protocol computing an approximation of~$f$ with
average-case error at most~$\epsilon$, and is denoted
by~$\sQ^*_{p}(f,\epsilon)$.
\end{definition}

Similarly, we say a protocol computes an
approximation of~$f$ with worst-case error at most~$\epsilon$
if there are quantum states~$\set{z_{xy} : x \in X, y \in Y, (x,y,z_{xy})
\in f}$ such that
\[
\max_{x \in X, y \in Y} \rP(w_{xy},z_{xy}) \quad \leq \quad \epsilon
\enspace.
\]

\begin{definition} 
The \emph{worst-case communication complexity\/} of~$f$ is defined as the minimum number of bits exchanged in an LOCC protocol computing an approximation of~$f$ with worst-case error at most~$\epsilon$, and is denoted by~$\sQ^*(f,\epsilon)$.
\end{definition}
 
Note that ``error'' here refers to the quality of approximation in the output state. The result of any probabilistic error made by the protocol is included in the output state, and hence this kind of error is reflected in the quality of approximation.

\subsection{Quantum information theory}
\label{sec:Asymptotic Information Theory}

Let $X$ be a register in quantum
state $\rho\in\sD(\cH)$. Then the \textit{von Neumann entropy}~$\rS(\rho)$ of
$X$ is defined as
\[
\rS(\rho)\quad\coloneqq \quad -\trace(\rho\log\rho)\enspace.
\]
Let $X$ and $Y$ be two registers in quantum states $\rho_{X}\in\sD(\cH)$
and $\sigma_{Y}\in\sD(\cH)$, respectively. The \emph{relative
entropy} denoted by $\rS(\rho_{X}\|\sigma_{Y})$ is defined as
\[
\rS(\rho_{X}\|\sigma_{Y}) \quad \coloneqq \quad 
\trace\left(\rho_{X}\log\rho_{X}-\rho_{X}\log\sigma_{Y}\right) 
\]
if~$\support(\rho)\subseteq \support(\sigma)$, and as~$\infty$ otherwise.
Suppose that~$\rho_{XY}\in\sD(\cH\otimes\cK)$ is the joint state of registers~$X$ and~$Y$, then the \textit{mutual information} of~$X$ and~$Y$ is defined as 
\[
\rI(X:Y)_{\rho}\quad\coloneqq\quad\rS(\rho_{X})+\rS(\rho_{Y})-\rS(\rho_{XY})\enspace,
\]
where~$\rho_X=\trace_Y(\rho_{XY})$ and~$\rho_Y=\trace_X(\rho_{XY})$.
When the register whose state is~$\rho$ is clear from the context,
we may omit it from the subscript of~$\rho$.
Similarly, when the state~$\rho$ of the registers~$XY$ is clear from 
the context, we may omit it from the subscript of~$\rI(X:Y)$.

For~$\rho,\sigma\in\sD(\cH)$, the \emph{observational
divergence\/}~\cite{jain_privacy_2002} between~$\rho$ and~$\sigma$ is defined as
\[
\rD(\rho\|\sigma)\quad\coloneqq\quad\sup\left\{
\trace(M\rho)\log\frac{\trace(M\rho)}{\trace(M\sigma)} \;:\;
0\leq M\leq\id,\trace(M\sigma)\neq0\right\} \enspace.
\]

Let~$\scE=\bigl( (p_{j},\rho_{j}):1\leq j\leq n \bigr)$ be an ensemble
of quantum states, i.e., $0\leq p_{j}\leq1$ for~$1\leq j\leq n$,
$\sum_{j=1}^{n}p_{j}=1$, and~$\rho_{j} \in \sD(\cH)$ are quantum states
over the same space. The \textit{Holevo
information} of $\scE$, denoted as~$\chi(\scE)$, is defined as
\[
\chi(\scE)\quad\coloneqq\quad\sum_{j=1}^{n}p_{j} \, \rS(\rho_{j}\|\rho)\enspace,
\]
where~$\rho$ is the ensemble average, i.e., $\rho=\sum_{j=1}^{n}p_{j}\rho_{j}$.
Similarly, we define the \textit{divergence information} of~$\scE$,
denoted as~$\rD(\scE)$, as
\[
\rD(\scE)\quad\coloneqq\quad\sum_{j=1}^{n}p_{j} \, \rD(\rho_{j}\|\rho)\enspace.
\]

Let~$S$ be a set, and~$Q:S\rightarrow\sD(\cH)$ be a function which
``encodes'' each~$x\in S$ as a quantum state. Let~$p$ be a probability
distribution over~$S$, and~$\rho_{AB}(p)$ be the bipartite
state~$\rho_{AB}(p) \coloneqq  \sum_x p_x \density{x}_A \otimes Q(x)_B$. We define the \emph{maximum
possible information} in~$Q$~\cite{jain_communication_2006},
denoted by~$\sT(Q)$, as 
\[
\sT(Q)\quad\coloneqq\quad\max_{p}\:\rI(A:B)_{\rho(p)}\enspace,
\]
where the maximum is taken over all probability distributions~$p$
over~$S$.

Note that for a classical-quantum state~$\rho_{AB}=\sum_{j=1}^{n}p_{j}\ketbra{j}j\otimes\rho_{j}$,
the mutual information of~$A$ and~$B$ is equal to
the Holevo information of the quantum ensemble~$\scE=\bigl(
(p_{j},\rho_{j}):1\leq j\leq n \bigr)$, i.e., $\chi(\scE)=\rI(A:B)$, and therefore~$\sT(Q)\geq\chi(\scE)$.

Most of the entropic quantities defined above
arise naturally in the analysis of information processing
tasks in the \emph{asymptotic\/} setting, i.e., when the available
resources may be used to jointly complete arbitrarily long sequences 
of tasks on independent, identically distributed (iid) inputs.
The asymptotic setting is an idealization that may not be realistic in
certain scenarios. More often, we are faced with single instances of a
task which we wish to accomplish with the fewest resources.
Recently, researchers have begun to formally study tasks in the non-iid or
\emph{one-shot\/} setting, and the entropic notions that arise therein.
Several one-shot entropic concepts have been implicit in traditional (iid) 
information theory and in communication complexity. 
For example, Jain, Radhakrishnan, and Sen implicitly studied the concept of
\emph{smooth max-relative entropy\/} in Ref.~\cite{jain_privacy_2002}.
However, non-asymptotic concepts were formalized only
later (see, e.g.,
Refs.~\cite{renner_security_2005,renner_smooth_2004,datta_min-_2009}).
In this work, we use one-shot entropic quantities to tightly 
characterize the communication complexity of remote state preparation.

Let~$\rho,\sigma\in\sD(\cH)$ be two quantum states. The \textit{max-relative
entropy} of~$\rho$ with respect to~$\sigma$ is defined as 
\[
\Dmax(\rho\|\sigma)\quad\coloneqq\quad\min\{\lambda:\rho\leq2^{\lambda}\sigma\}\enspace,
\]
when~$\support(\rho) \subseteq \support(\sigma)$, and is~$\infty$ otherwise~\cite{datta_min-_2009}.
This notion captures how two states~$\rho,\sigma$ behave relative to each
other under the application of a measurement. 
For a bipartite quantum state~$\rho_{AB}\in\sD(\cH'\tensor\cH)$,
the \textit{max-information} part~$B$ has about 
part~$A$~\cite{berta_quantum_2011} is defined as
\[
\Imax(A:B)_\rho
\quad\coloneqq\quad\min_{\sigma\in\sD(\cH)}\Dmax(\rho_{AB} \,\|\, \rho_{A}\otimes\sigma_B)\enspace.
\]
Note that this quantity is asymmetric with respect to the parts~$A$ and~$B$.
As for mutual information, we include the state as a subscript only when
it is not clear from the context. 
The \emph{smoothed\/} versions of these quantities come into play when
approximations are allowed in the tasks at hand. \textit{Smooth
max-relative entropy\/} is defined as
\[
\Dmax^{\epsilon}(\rho\|\sigma)\quad\coloneqq\quad\min_{\tilde{\rho}\in\sB^{\epsilon}(\rho)}\Dmax(\tilde{\rho}\|\sigma)\enspace,
\]
and \textit{smooth max-information\/} is defined as
\[
\Imax^{\epsilon}(A:B)_{\rho}\quad\coloneqq\quad\min_{\tilde{\rho}\in\sB^{\epsilon}(\rho)}\Imax(A:B)_{\tilde{\rho}}\enspace.
\]
There are several ways to define max-information using max-relative
entropy~\cite{ciganovic_smooth_2013}. We choose the above definition in
this work since it can be used to characterize average-case communication
complexity of the remote state preparation problem. 

The following are some properties of max-information we use.
Both the exact and smooth versions of this quantity are monotonic under 
the application of a quantum channel~\cite{berta_quantum_2011}.
\begin{proposition}[Monotonicity under quantum channels]
\label{thm-monotonicity-mi}
Let~$\Phi:\rL(\cH')\rightarrow\rL(\cK)$ be a quantum channel,
$\rho_{AB}$ a bipartite sub-normal state over~$\cH'\tensor\cH$,
$\sigma_{AB}\in\sD(\cH'\tensor\cH)$
a bipartite quantum state, and~$\epsilon \in [0,1]$. Then
\begin{align*}
\Imax(A':B)_{\rho'}
    \quad & \leq \quad \Imax(A:B)_{\rho}\enspace, \qquad
                 \textrm{and} \\
\Imax^\epsilon(A':B)_{\sigma'}
    \quad & \leq \quad \Imax^\epsilon(A:B)_{\sigma}\enspace,
\end{align*}
where~$A',B$ denote two parts
of the states~$\rho'_{A'\!B} \coloneqq  (\Phi\tensor\id)(\rho)$
and~$\sigma'_{A'\!B} \coloneqq  (\Phi\tensor\id)(\sigma)$.
\end{proposition}

For a classical-quantum state~$\rho_{AB}$, the value of smooth
max-information is achieved by a classical-quantum state~$\rho'_{AB}$
that is~$\epsilon$-close to~$\rho_{AB}$. A proof is included in
Appendix~\ref{App-proof of properties}.

\begin{proposition}
\label{thm-classical-smi}
Let~$\rho_{AB}\in\sD(\cH'\tensor\cH)$
be a bipartite quantum state that is classical on~$A$. For any~$\epsilon\ge0$, there exists~$\rho'_{AB}\in\sB^{\epsilon}(\rho_{AB})\cap\sD(\cH'\tensor\cH)$
classical on~$A$ such that 
\[
\Imax^{\epsilon}(A:B)_{\rho}\quad=\quad\Imax(A:B)_{\rho'}\enspace.
\]
\end{proposition}

Smooth max-information satisfies the Asymptotic Equipartition Property,
as proven by Berta, Christandl, and Renner~\cite{berta_quantum_2011}.
Let~$\rH$ denote the binary entropy function~$\rH(\alpha) \coloneqq  - \alpha
\log \alpha - (1-\alpha)\log (1-\alpha)$.

\begin{theorem}[Quantum Asymptotic Equipartition property]
\label{thm-QAEP-Imax}
Let $\epsilon>0$, $n$ an integer such that~$n \geq 2(1-\epsilon^2)$,
and $\rho_{AB}\in \sD(\cH_{AB})$. Then
\begin{align}
\label{eq-QAEP-lb}
\rI(A:B)_\rho - \frac{3}{n} \, \rH(\epsilon)
- 2\epsilon \log(\abs{A}\abs{B})
    \quad \leq \quad \frac{1}{n} \, \Imax^{\epsilon}(A:B)_{\rho^{\otimes n}}
\enspace,
\end{align}
and
\begin{align}
\label{eq-QAEP-ub}
\frac{1}{n} \, \Imax^{\epsilon}(A:B)_{\rho^{\otimes n}}
    \quad \leq \quad \rI(A:B)_\rho
     + \frac{\xi(\epsilon)}{\sqrt{n}}
     - \frac{2}{n} \log{\frac{\epsilon^2}{24}} \enspace, 
\end{align}
where~$\xi(\epsilon) = 8\sqrt{13-4\log \epsilon} \,
(2+\frac{1}{2}\log \abs{A})$. Therefore,
\[
\lim_{\epsilon\rightarrow0} \lim_{n\rightarrow\infty} ~ \frac{1}{n}
\, \Imax^{\epsilon}(A:B)_{\rho^{\otimes n}}
    \quad = \quad \rI(A:B)_\rho \enspace.
\]
\end{theorem}

For $\epsilon \in [0,1)$, the \textit{$\epsilon$-hypothesis testing relative
entropy\/}~\cite{wang_one-shot_2012} of two quantum states~$\rho,\sigma
\in \sD(\cH)$ is defined as
\[
\Dh^{\epsilon}(\rho\|\sigma)\quad\coloneqq\quad-\log\frac{\upbeta^{\epsilon}(\rho\|\sigma)}{1-\epsilon}\enspace,
\]
where
\begin{align}
\label{eq:def-betaEpsilon}
\upbeta^{\epsilon}(\rho\|\sigma)
    \quad \coloneqq \quad \inf ~ \{ \langle Q,\sigma\rangle
    \;|\; 0\leq Q\leq\id ~ \textrm{and} ~ \langle Q,\rho\rangle\geq1-\epsilon\} 
    \enspace.
\end{align}
The infimum in the above definition is always achieved
and~$\upbeta^\epsilon(\rho\|\sigma)$ is between~$0$ and~$1$. In this
definition, we interpret~$(Q,\id-Q)$ as a measurement for
distinguishing~$\rho$ from~$\sigma$, i.e., as a strategy in
\emph{hypothesis testing\/}. So~$\upbeta^\epsilon(\rho\|\sigma)$
corresponds to the minimum probability of \emph{incorrectly\/} 
identifying~$\sigma$ when~$\rho$ is identified \emph{correctly\/} with 
probability at least~$1-\epsilon$.
This one-shot entropic quantity has been studied for a 
long time either implicitly (see, e.g., Refs.~\cite{hiai1991,ogawa2000strong})
or explicitly, albeit without giving it a name
(see, e.g., Refs.~\cite{BD10,BD11}). It also arises in the 
context of \textit{channel coding\/}~\cite{hayashi2002general,wang_one-shot_2012} and other tasks~\cite{Hayashi17}.

The error in hypothesis testing may only increase under the action of a
quantum channel. This has been known for some time; see, e.g.,
Ref.~\cite[Eq.~(44)]{bjelakovic2003quantum} for a proof.

\begin{proposition}[Data Processing Inequality]
\label{thm-DPI-beta}
Let $\rho,\sigma\in\sD(\cH)$ for some Hilbert space $\cH$, and $\Phi:\sL(\cH)\rightarrow\sL(\cK)$
be a quantum channel. Then
\[
\upbeta^{\epsilon}(\rho\|\sigma)\quad\leq\quad\upbeta^{\epsilon}(\Phi(\rho)
\,\|\, \Phi(\sigma))\enspace.
\]
\end{proposition}
The following two properties have been proved implicitly by Matthews and
Wehner~\cite{matthews_finite_2012}. For completeness, we include their
proofs in Appendix~\ref{App-proof of properties}. 

Hypothesis testing error satisfies a restricted form of joint convexity in its two arguments.

\begin{proposition} 
\label{thm-covexity of Be in p}
Let~$\rho_{AB}(p)\in\sD(\cH_{A}\otimes\cH_{B})$ be a state classical
on~$A$ such that the distribution on~$A$ is given by the probability
vector $p$. Let~$\rho_{A}(p)=\trace_{B}(\rho_{AB}(p))$,
and~$\sigma\in\sD(\cH_{B})$ be a quantum state on Hilbert
space $\cH_{B}$. Then the function $\upbeta^{\epsilon}(\rho_{AB}(p)
\,\|\, \rho_{A}(p)\otimes\sigma)$
is convex with respect to~$p$.
\end{proposition}
Hypothesis testing error is concave in its second argument.

\begin{proposition}
\label{thm-concavity of Be in sigma}
For any fixed quantum state~$\rho\in\sD(\cH)$, the function
$\upbeta^{\epsilon}(\rho\|\sigma)$ is a concave function with 
respect to $\sigma$.
\end{proposition}

It turns out that hypothesis testing relative entropy is closely 
related to smooth max-relative entropy, as captured by the following
theorem.
\begin{theorem}[\cite{dupuis_generalized_2013,tomamichel_hierarchy_2013}]
\label{thm-Dmax-Dh} 
Let~$\rho,\sigma\in\sD(\cH)$ be two quantum states in Hilbert space~$\cH$.
For any~$\epsilon \in (0,1)$ and~$\delta \in (0,\epsilon)$, 
the following inequalities hold:
\begin{align}
\label{eq:Dmax<Dh}
\Dmax^{\sqrt{2(1-\epsilon)}}(\rho\|\sigma) 
    \quad \leq & \quad \Dh^{\epsilon}(\rho\|\sigma) \enspace ,
        \qquad \textrm{and} \\
\label{eq:Dmax>Dh}
\Dmax^{\sqrt{1-\epsilon}}(\rho\|\sigma) 
    \quad \geq & \quad \Dh^{\epsilon-\delta}(\rho\|\sigma)
        - \log \frac{\epsilon(1-\epsilon+\delta)}{\delta^{3}}
        - 3\log3 \enspace .
\end{align}
\end{theorem}

\subsection{The minimax theorem}

The minimax theorem is a powerful result that provides conditions under
which switching the order of minimization and maximization in certain
optimization problems does not change the optimum.

\begin{theorem}[\cite{osborne_course_1994}]
\label{thm-minimax theorem}
Let~$n$ be a positive integer, and~$A_{1}, A_{2}$ be non-empty, 
convex and compact subsets of~$\reals^{n}$.
Let $f:A_{1}\times A_{2}\rightarrow\reals$
be a continuous function such that
\begin{enumerate}
\item $\forall a_{2}\in A_{2}$, the set $\{a_{1}\in A_{1} :
(\forall a'_{1}\in A_{1}) \; f(a_{1},a_{2})\geq f(a'_{1},a_{2})\}$
is convex.
\item $\forall a_{1}\in A_{1}$, the set $\{a_{2}\in A_{2}:(\forall
a'_{2}\in A_{2}) \; f(a_{1},a_{2})\leq f(a_{1},a'_{2})\}$
is convex.
\end{enumerate}
Then 
\[
\max_{a_{1}\in A_{1}}\min_{a_{2}\in A_{2}}f(a_{1},a_{2})\quad=\quad\min_{a_{2}\in A_{2}}\max_{a_{1}\in A_{1}}f(a_{1},a_{2})\enspace.
\]

\end{theorem}

\subsection{Remote state preparation}
\label{part:Approximate-remote-state}
 
Let~$S$ be a finite, non-empty set, and
let~$Q:S\rightarrow\sD(\cH)$ be a function that maps each element~$x\in
S$ to a quantum state~$Q(x)$ over the Hilbert space~$\cH$. Recall that
remote state preparation, denoted as~$\RSP(S,Q)$, is a communication
task in which one party, Alice, is given an input~$x\in S$, and engages
in an LOCC protocol with another party, Bob, so that Bob is able to
prepare~$Q(x)$. The function~$Q$ is known to both parties. In the
approximate remote state preparation, we allow Bob to prepare an
approximation~$\sigma_x \in \sD(\cH)$ to~$Q(x)$. We consider two notions of error in approximation: worst case and average case. Let~$\epsilon \in [0,1]$, and let~$p$ be a probability distribution on~$S$.
We say a protocol for~$\RSP(S,Q)$ makes worst-case error~$\epsilon$
if~$\rP(\sigma_x, Q(x))\leq \epsilon$ for each~$x\in S$. We say a
protocol for~$\RSP(S,Q)$ makes average-case error~$\epsilon$ w.r.t.\ the
distribution~$p$ over~$S$ if the purified distance between the ideal and
actual joint input-output states is at most~$\epsilon$. By the
definitions of purified distance and fidelity, this condition is
equivalent to
\[
\sum_{x\in S} p_x \, \rF(\sigma_x,Q(x)) 
    \quad \geq \quad \sqrt{1 - \epsilon^2} \enspace.
\]

In Sections~\ref{sec:Average-case-error-communication} and~\ref{sec:Worst-case-error-communication}, we characterize the communication complexity of this
problem for the two different kinds of approximation. We emphasize that Alice and Bob communicate with
a noiseless classical channel, they have access to an arbitrarily large
amount of entanglement of their choice,
and they have unlimited computational power. 

A straightforward protocol for approximate remote state preparation is
as follows. Alice sends her input~$x$ directly to Bob and Bob creates
the desired state~$Q(x)$. Thus Bob prepares the target state with zero error
($\epsilon=0$) using~$\lceil\log(n+1)\rceil$
bits of classical communication, where~$n=\abs{S}$. 

Jain, Radhakrishnan, and Sen~\cite{jain_privacy_2002,jain_prior_2005}
proposed the following, potentially more efficient protocol, which we
call the JRS protocol in the sequel.
Let~$\cK$ be a Hilbert space with~$\dim(\cK)\geq \dim(\cH)$ and~$\{\sigma_{x}\}_{x\in S}\subseteq\sD(\cH)$
be a set of quantum states such that for all~$x\in S$,~$\rP(\sigma_{x},Q(x))\leq \delta$ for some~$\delta\in[0,1]$.
Suppose that for some~$\lambda\in[0,\infty)$ and some~$\sigma\in\sD(\cH)$,
we have 
\begin{equation}
\sigma_{x}\quad\leq\quad2^{\lambda}\sigma\quad\quad\mathrm{for\: all}\quad x\in S\enspace.\label{eq:RSP protocol}
\end{equation}
This can be rewritten for a fixed~$x \in S$ as
$$\sigma \quad = \quad 2^{-\lambda}\sigma_x+(1-2^{-\lambda})\xi_x\enspace,$$
where~$\xi_x\in\sD(\cH)$ is a quantum state. Let~$\ket{v_{x}}\in\cK\tensor\cH$ be a purification of~$\sigma_x$ in the Hilbert space~$\cK\otimes\cH$, and~$\ket{u_{x}}\in\cK\tensor\cH$ be a purification of~$\xi_x$. Then
\[
\ket{w_{x}}\quad=\quad\sqrt{2^{-\lambda}}\:\ket{0}\ket{v_{x}}\quad+\quad\sqrt{1-2^{-\lambda}}\:\ket{1}\ket{u_{x}}\enspace,
\]
is a purification of~$\sigma$. 
Let~$\ket{w}$ be an arbitrary but fixed purification of~$\sigma$
in~$\complex^{2}\tensor\cK\tensor\cH$. By the unitary equivalence of
purifications, there is a unitary operation~$U_{x}$ on
the space~$\complex^{2}\tensor\cK$ which transforms~$\ket{w}$ to~$\ket{w_{x}}$.
We are ready to describe the JRS protocol.

\textbf{JRS Protocol:}
Alice and Bob agree on a parameter~$t$, that depends on the quality of
approximation they desire. Initially, Alice and Bob share~$t$ copies 
of the quantum state~$\ket{w}$. The registers corresponding to Hilbert
spaces~$\complex^{2}$ and~$\cK$ in the~$i$th copy of~$\ket{w}$ are
called~$C_i$ and~$K_i$, respectively, and are held by Alice. The register
corresponding to the Hilbert space~$\cH$ is called~$H_i$ and is held by Bob.
\begin{enumerate}

\item
On getting input~$x$, Alice performs the unitary operation~$U_{x}$ on
registers~$C_i K_i$ for each~$i \in [t]$. This transforms all copies
of~$\ket{w}$ to copies of~$\ket{w_{x}}$. Then she measures
the register~$C_i$ for all~$i \in [t]$.
If at least one of the measurement outcomes, say the~$j$th,
is equal to zero, she sends the index~$j$ to Bob,
using~$\ceil{\log(t+1)}$ bits. (She may choose to
send any such index.)
Otherwise, if the outcomes of all~$t$ measurements are equal to one,
she sends~$0$ to Bob.

\item On receiving an integer~$k$, where~$0 \le k \le t$, 
Bob outputs the state in register~$H_k$ if~$k
\in [t]$, and outputs the maximally mixed state over~$\cH$ if~$k = 0$.

\end{enumerate}

The output of this protocol is~$\frac{\id}{\abs{H}}$ with
probability~$\left(1-2^{-\lambda} \right)^{t}$ and~$\sigma_{x}$ with the
remaining probability.
Hence, the output state is 
\[
\tilde{\sigma}_{x} \quad = \quad
    \left( 1 - \left( 1 - 2^{-\lambda} \right)^{t} \right) \; \sigma_{x}
    + \left( 1 - 2^{-\lambda} \right)^{t} \; \frac{\id}{\abs{H}} \enspace.
\]
By choosing the approximation parameter~$\delta$ small enough 
and~$t$ large enough, Bob produces a  state~$\tilde{\sigma}_{x}$ with
the desired accuracy. We use this protocol
to give upper bounds on the worst-case error and average-case communication
complexity of~$\RSP(S,Q)$.

\section{Average-case communication complexity}
\label{sec:Average-case-error-communication}

Let~$p$ be a probability distribution over~$S$ and~$\sQ_{p}^{*}(\RSP(S,Q),\epsilon)$
denote the average-case entanglement-assisted communication complexity
of approximate remote state preparation (ARSP), with respect to~$p$, and with
(average) error at most~$\epsilon$. We characterize this quantity in terms of smooth max-information, a one-shot analogue of mutual information.

\subsection{An upper bound}

First, we show that the average-case communication complexity with
error~$\epsilon$ of ARSP is bounded above essentially by~$\Imax^{\delta}(A:B)_{\rho(p)}$,
where~$\rho(p)$ is the ideal joint state of Alice's input and Bob's
output, and~$\delta\in \Theta(\epsilon)$. To do so, we use the JRS
protocol described in Section~\ref{part:Approximate-remote-state}.

\begin{theorem} \label{thm-Ub} For any finite set~$S$, function~$Q:S\rightarrow\sD(\cH)$,
and~$\epsilon\in(0,1]$, let~$p$ be a probability distribution
over~$S$ and~$\rho_{AB}(p)\in\sD(\cH'\tensor\cH)$ be the bipartite
classical-quantum state~$\rho_{AB}(p)=\sum_{x\in S}{p_{x}\density{x}_{A}\tensor Q(x)_{B}}$.
Then 
\[
\sQ_{p}^{*}(\RSP(S,Q),\epsilon)\quad\leq\quad\Imax^{\delta}(A:B)_{\rho(p)}+\log_{2}
\ln\frac{8}{\epsilon^{2}} + 2\enspace,
\]
where~$\delta=\epsilon/2\sqrt{2}$.
\end{theorem}

\begin{proof} Fix some~$\epsilon\in(0,1]$, and let~$\lambda$ be
equal to~$\Imax^{\delta}(A:B)_{\rho(p)}$ with~$\delta$ as in the statement of the theorem. By Proposition~$\ref{thm-classical-smi}$,
there exist quantum states~$\rho'_{AB}\in\sB^{\delta}(\rho_{AB})$ and~$\sigma_{B}\in\sD(\cH)$
such that~$\rho'_{AB}\leq2^{\lambda}\rho'_{A}\tensor\sigma_{B}\enspace,$where~$\rho'_{AB}=\sum_{x}q_{x}\ketbra{x}{x}\tensor\sigma_{B}^{x}$
with~$\sum_{x}q_{x}=1$ and~$\sigma_{B}^{x}\in\sD(\cH)$, and~$\rho'_{A}=\sum_{x}q_{x}\ketbra{x}{x}$.
Then
\begin{equation}
\sigma_{B}^{x}\quad\leq\quad 2^{\lambda}\sigma_{B}\enspace,\label{eq:ineqforUpperbounAvg}
\end{equation}
for all~$x\in S$ with~$q_x\neq0$. For each~$x\in S$ with
$q_{x}=0$, we assume, w.l.o.g., that~$\sigma_{B}^{x}=\sigma_{B}$.
Inequality~\eqref{eq:ineqforUpperbounAvg} is in the form of inequality~\eqref{eq:RSP protocol} and therefore
we may execute the JRS protocol with a suitable choice of parameter~$t$.
Initially, Alice and Bob share~$t$ copies of entangled state~$\ket{w}$,
where~$\ket{w}$ is a purification of~$\sigma_B$.
Alice gets input~$x$ with probability~$p_x$.
They perform the protocol for approximating state~$\sigma_B^x$ from~$\sigma_B$. The final joint state of Alice's input and Bob's output is 
\[
\tilde{\rho}_{AB}\quad=\quad\sum_{x\in S}p_{x}\density{x}\otimes\tilde{\sigma}_{B}^{x}\enspace,
\]
where
\[
\tilde{\sigma}_{B}^{x} \quad = \quad 
    \left( 1 - \left( 1 - 2^{-\lambda} \right)^{t} \right)
    \, \sigma_{B}^{x} + \left(1-2^{-\lambda} \right)^{t} \frac{\id}{\dim(\cH)}
    \enspace.
\]
Therefore, 
\begin{align*}
\rF(\tilde{\rho}_{AB},\rho'_{AB})\quad= & \quad\rF \! \left(\sum_{x\in
S}p_{x}\density{x}\otimes\tilde{\sigma}_{B}^{x}, \; \sum_{x\in S}q_{x}\density{x}\otimes\sigma_{B}^{x}\right)\\
\geq & \quad\big(1-(1-2^{-\lambda})^{t}\big) \; \rF \! \left(\sum_{x\in
S}p_{x}\density{x}\otimes\sigma_{B}^{x}, \; \sum_{x\in S}q_{x}\density{x}\otimes\sigma_{B}^{x}\right)\\
= & \quad\big(1-(1-2^{-\lambda})^{t}\big)\sum_{x\in S}\sqrt{p_{x}q_{x}}\\
\geq & \quad\left(1-(1-2^{-\lambda})^{t}\right)\sqrt{1-\delta^{2}}\enspace,
\end{align*}
where the first inequality follows from the joint concavity of fidelity.
The last inequality follows from monotonicity under quantum
channels:
\begin{align*}
\sum_{x\in S}\sqrt{p_{x}q_{x}}\quad= & \quad\rF(\rho'_{A},\rho_{A})\quad
\geq \quad\rF(\rho'_{AB},\rho_{AB})\enspace.
\end{align*}
In addition, by Proposition~\ref{eq:triangle-ineq-fidelity}, 
\begin{align*}
\rF(\tilde{\rho}_{AB},\rho_{AB})\quad\geq & \quad\rF(\tilde{\rho}_{AB},\rho'_{AB})^{2}+\rF(\rho_{AB},\rho'_{AB})^{2}-1\\
\geq & \quad\big(1-(1-2^{-\lambda})^{t}\big)^2(1-\delta^{2})+(1-\delta^{2})-1\\
\geq & \quad\sqrt{1-\epsilon^{2}}\enspace,
\end{align*}
where the last inequality is derived using inequalities~$\ln(1-x)\leq-x$
and~$\sqrt{1-x} \leq 1-\frac{x}{2}$, which hold for~$x\in[0,1)$, and the
parameter values~$\delta=\epsilon/2\sqrt{2}$ and~$t=\ceil{2^\lambda \ln
\tfrac{8}{\epsilon^2}}$.
Since~$\rF(\tilde{\rho}_{AB},\rho_{AB})=\sum_{x\in S}p_x \, \rF(\tilde{\sigma}_x, Q(x))$, the protocol
has average-case error at most~$\epsilon$. 

The communication cost of this protocol is~$\lceil \log(t+1) \rceil$. So the communication complexity of approximate
remote state preparation with average-case error~$\epsilon$ is
\begin{align*}
\sQ_{p}^*(\RSP(S,Q),\epsilon)\quad \leq & \quad \left\lceil \log(t+1)
\right\rceil \quad \leq \quad \lambda+\log_{2} \ln \frac{8}{\epsilon^2} +2\enspace,
\end{align*}
as required.
\end{proof}

We have not attempted to optimize the upper bound derived above.
It is possible that the parameter~$\delta$ and the~$\epsilon$-dependent
additive term be improved further.

\subsection{A lower bound}

Next, we show that the average-case communication complexity of any protocol
for approximate remote state preparation is bounded from below by~$\Imax^{\epsilon}(A:B)_{\rho(p)}$.
In order to do this, we strengthen a property of
smooth max-information due to Berta, Christandl, and Renner~\cite[Lemma
B.12]{berta_quantum_2011}, in the case of a tripartite state~$\rho_{M\!
AB}$ that is classical on~$M$.

\begin{lemma} \label{thm-upper-dp} Let~$\epsilon\geq 0$ and~$\rho_{M\! AB}\in\sD(\cM\tensor\cH'\tensor\cH)$
be any tripartite quantum state over registers~$M$,~$A$ and~$B$ such
that~$\rho$ is classical on~$M$. Then
\[
\Imax^{\epsilon}(A:M\! B)\quad\leq\quad\Imax^{\epsilon}(A:B)+\log\abs{M}\enspace.
\]

\end{lemma}

\begin{proof} Fix~$\sigma_{B}\in\sD(\cH)$ and~$\tilde{\rho}_{AB}\in\sB^{\epsilon}(\rho_{A\! B})$
such that~$\Imax^{\epsilon}(A:B)=\Dmax(\tilde{\rho}_{A\! B} \,\|\,
\tilde{\rho}_{A}\otimes\sigma_{B})$. Let~$\lambda$ denote this
max-relative entropy, i.e.,~$\lambda$ is the minimum non-negative real number for which~$\tilde{\rho}_{A\! B}\leq2^{\lambda}\tilde{\rho}_{A}\otimes\sigma_{B}$.
Then
\begin{align}\label{eq:First_in_lemma}
\frac{\id}{\abs{M}}\otimes\tilde{\rho}_{AB}\quad\leq\quad2^{\lambda} \; \frac{\id}{\abs{M}}\otimes\tilde{\rho}_{A}\otimes\sigma_{B}\enspace.
\end{align}

By Proposition~$\ref{cor-uhlmann}$, there exists some extension~$\rho'_{M\!
AB}$ of~$\tilde{\rho}_{AB}$ such that~$\rho'_{M\!
AB} \in\sB^{\epsilon}(\rho_{M\! AB})$. By construction,
we have~$\trace_{M}(\rho'_{M\! AB})=\tilde{\rho}_{AB}$. Consider
the quantum-to-classical channel~$\Phi:\rL(\cM)\rightarrow\rL(\cM)$ defined by
\[
\Phi(X)\quad=\quad\sum_{i}\bra{e_i}X\ket{e_i}\ket{e_i}\!\bra{e_i}
\]
for all~$X\in\rL(\cM)$, where~$\{\ket{e_i}\}$ is the standard orthonormal
basis for the Hilbert space~$\cM$.
The state~$(\Phi\otimes\id)(\rho'_{M\! AB})$ is classical
on~$M$, and is an extension of~$\tilde{\rho}_{AB}$.
Define~$\tilde{\rho}_{M\! AB} \coloneqq (\Phi\otimes\id)(\rho'_{M\!
AB})$.
Since~$\rho'_{M\! AB}\in\sB^{\epsilon}(\rho_{M\! AB})$, by monotonicity
of fidelity under quantum channels and because~$\rho_{M\! AB}$
is classical on~$M$, we have~$\tilde{\rho}_{M\!
AB}\in\sB^{\epsilon}(\rho_{M\! AB})$. So~$\tilde{\rho}_{M\! AB}$ may
be written as 
\[
\tilde{\rho}_{M\!
AB}\quad=\quad\sum_{i}{\gamma_{i}\ket{e_i}\!\bra{e_i}\otimes\sigma_{AB}^{i}}\enspace,
\]
where all~$\sigma_{AB}^{i}$ are normalized and~$\sum_{i}{\gamma_{i}}\leq1$.
We have~$\tilde{\rho}_{M\! AB}\leq\id_{M}\otimes\tilde{\rho}_{AB}$.
Combining this with Equation~(\ref{eq:First_in_lemma}), we can conclude that 
\[
\tilde{\rho}_{M\! AB}\quad\leq\quad2^{\lambda}|M|\left(\frac{\id_{M}}{|M|}\otimes\tilde{\rho}_{A}\otimes\sigma_{B}\right)
\]
and consequently, 
\[
\Dmax\left(\tilde{\rho}_{M\! AB} \left\| ~
\frac{\id_{M}}{|M|}\otimes\tilde{\rho}_{A}\otimes\sigma_{B} \right. \right)
    \quad\leq\quad\lambda+\log{|M|}\enspace.
\]
By the definition of smooth max-information, this implies that 
\[
\Imax^{\epsilon}(A:M\! B)\quad\leq\quad\lambda+\log{|M|}\enspace ,
\]
as required.
\end{proof}

\begin{remark}
The above lemma could alternatively be derived from an analogous inequality 
for~\emph{$\alpha$-R\'enyi mutual
information\/}~\cite[Equation~(2.25)]{leditzky2016strong}). 
Taking the limit as~$\alpha \rightarrow \infty$ gives us the 
inequality for max-information (i.e., for~$\epsilon = 0$). We may extend
this to any~$\epsilon \geq 0$ by smoothing arguments similar to
those in the above proof.
\end{remark}

Using this lemma, we bound the average-case communication complexity
of~$\RSP(S,Q)$ from below.

\begin{theorem} \label{thm-lb} For any finite set~$S$, function~$Q:S\rightarrow\sD(\cH)$, and probability distribution~$p$
over~$S$, let~$\rho(p)$ be the bipartite quantum state
$$\rho(p) \quad = \quad \sum_{x\in S}{p_{x}\density{x}_{A}\tensor Q(x)_{B}}\enspace .$$
For any~$\epsilon\in[0,1]$, we have
\[
\sQ_{p}^{*}(\RSP(S,Q),\epsilon)\quad\geq\quad\Imax^{\epsilon}(A:B)_{\rho(p)}\enspace.
\]
\end{theorem}

\begin{proof}
In this proof we follow the notation and convention described in
Section~\ref{sec-LOCC protocols}. Consider a~$k$-round LOCC protocol~$\Pi$ for~$\RSP(S,Q)$ with average-case error~$\epsilon$. Suppose Bob sends the first message, and the joint state in Alice and Bob's registers (excluding the input register~$A$) after the message is~$\phi$. As Bob receives no input, the joint state~$\phi$ is known to both parties. Hence, the rest of the protocol can be considered as a \emph{new\/} LOCC protocol, with the same output, in which the initial shared state of parties is~$\phi$, \emph{and\/} Alice starts the protocol. The communication cost of this new protocol is less than the communication cost of the original one. Therefore, it suffices to show the lower bound for protocols in which Alice starts.

Let~$A$ be Alice's input register, and~$Y_{i} \coloneqq P_{i}V_{i}$ and~$Z_{i}\coloneqq Q_{i}W_{i}$ be Alice's and Bob's classical-quantum private registers, respectively, after the~$i$th round of the protocol for~$i\geq 0$.
Initially,~$A$ and~$Z_{0}$ are independent, and so
\begin{align}
\Imax(A:Z_{0}) \quad = \quad 0\enspace. \label{eq:initial-2-way}
\end{align}

Consider the~$i$th round of a two-way LOCC protocol. The communication in each round is either from Alice to Bob (for odd~$i$) or from Bob to Alice (for even~$i$).
 
\paragraph{Odd round~$i$:} In this case, Alice measures her private qubits~$V_{i-1}$ controlled by~$P_{i-1}$ and~$A$. She includes the outcome of her measurement~$M_{i}$ in the register~$P_{i}$ (recall that $P_{i}=
P_{i-1} M_{i}$), and sends a copy of~$M_{i}$ to Bob using~$m_{i}
\coloneqq \left\lceil \log (\abs{M_{i}}+1) \right\rceil$ bits of communication. Then Bob includes the received message~$M_{i}$ in~$Q_{i}$ (recall
that $Q_{i}=Q_{i-1}M_{i}$). Thus,
\begin{align}
\nonumber
\Imax(A:Z_{i}) \quad \leq &\quad \Imax(A: Z_{i-1})+\log \abs{M_{i}}& (\mathrm{by\ Lemma~\ref{thm-upper-dp}})\\
 \leq &\quad \Imax(A:Z_{i-1})+m_i\enspace. & \label{eq:odd-round}
\end{align}

\paragraph{Even round~$i$:} In this case, Bob measures his private qubits~$W_{i-1}$ controlled by~$Q_{i-1}$. He includes the outcome of his measurement~$M_i$ in~$Q_{i}$, and sends a copy of~$M_i$ to Alice using~$m_i=\left\lceil \log (\abs{M_{i}}+1) \right\rceil$ bits of communication. Alice includes the received message in~$P_{i}$. Thus,
\begin{align}
\Imax(A:Z_{i}) \quad \leq & \quad \Imax(A:Z_{i-1})\enspace. &
(\mathrm{by\ Proposition~\ref{thm-monotonicity-mi}}) \label{eq:even-round}
\end{align}

Combining Eqs.~\eqref{eq:odd-round} and~\eqref{eq:even-round} recursively, we get
\[
\Imax(A:Z_{k})\quad\leq\quad\Imax(A:Z_{0})+\sum_{\substack{1\leq i \leq k \\ i \ \mathrm{odd}}}{m_{i}}\quad =\quad \sum_{\substack{1\leq i \leq k \\ i \ \mathrm{odd}}}{m_{i}} \enspace,
\]
after~$k$ rounds of communication. Let~$m \coloneqq \sum_{1 \le i \le k,
\textrm{ odd}} m_i$. At the end of the protocol, Bob applies a quantum
channel on his register~$Z_{k}$  to get the output~$B$. By
monotonicity of max-information (Proposition~\ref{thm-monotonicity-mi}), we have 
\[
\Imax(A:B)_{\rho'(p)}\quad\leq\quad m\enspace,
\]
where~$\rho'(p)=\sum_x p_x\density{x}\otimes\sigma_x$ is the bipartite quantum state of registers~$AB$, and~$m$ is the number of bits of communication from Alice to Bob. In addition, protocol~$\Pi$ guarantees that~$\rho'(p)$
is within purified distance~$\epsilon$ of~$\rho(p)$. Therefore, we conclude the theorem.
\end{proof}

\section{Worst-case communication complexity\label{sec:Worst-case-error-communication}}

In this section, we characterize the
worst-case communication complexity of remote state preparation,
denoted as~$\sQ^{*}(\RSP(S,Q),\epsilon)$, in terms of smooth
max-relative entropy. 

\subsection{An upper bound}\label{subsec:ub-worst-case}

We show that for some fixed~$\epsilon\in(0,1]$, the worst-case communication complexity
of the approximate remote state preparation problem is bounded from above essentially by
$$\min_{\sigma\in\sD(\cH)}\max_{x\in S}\ \Dmax^{\delta}(Q(x)\|\sigma)\enspace,$$
where~$\delta\in\Theta(\epsilon)$. 
As for the average case, we utilize the JRS protocol presented in Section~\ref{part:Approximate-remote-state}.

\begin{theorem}\label{thm-worst-upperbound} Let~$S$ be a non-empty finite
set,~$Q:S\rightarrow\sD(\cH)$ be a function from~$S$ to the set
of density operators in the Hilbert space~$\cH$, and~$\epsilon\in[0,1]$.
Then
\[
\sQ^{*}(\RSP(S,Q),\epsilon)\quad\leq\quad\min_{\sigma\in\sD(\cH)}\max_{x\in S}\ \Dmax^{\delta}(Q(x)\|\sigma)+\log_2(1+\epsilon^{2})+\log_{2}\ln\frac{2}{\epsilon^{4}}+2\enspace,
\]
where~$\delta=\frac{\epsilon}{\sqrt{1+\epsilon^{2}}}$.

\end{theorem}

\begin{proof} Let~$\alpha \coloneqq \min_{\sigma\in\sD(\cH)}\max_{x\in S}\Dmax^{\delta}(Q(x)\|\sigma)$
and~$\sigma'$ be the quantum state for which the minimum is achieved,
i.e.,~$\alpha=\max_{x\in S}\Dmax^{\delta}(Q(x)\|\sigma')$. By definition, for
all~$x\in S$ there exists some~$\sigma_{x}\in\sB^{\delta}(Q(x))$
such that
\[
\sigma'\quad\geq\quad2^{-\alpha}\sigma_{x}\enspace.
\]
Since~$\rP(\sigma_{x},Q(x))\leq\delta$, we have~$\rF(\sigma_{x},Q(x))^2 \geq 1-\delta^2$. So, by Proposition~\ref{eq:fidelity-trace},~$\trace(\sigma_{x}) \geq 1-\delta^2 = \frac{1}{1+\epsilon^{2}}$
for all~$x\in S$. For each~$x\in S$, define~$\rho_{x}\coloneqq\frac{\sigma_{x}}{\trace(\sigma_{x})}$. Then for all~$x\in S$,~$\rho_{x}$ is a quantum state~$\delta$-close
to~$Q(x)$, i.e.,~$\rho_{x}\in\sB^{\delta}(Q(x))\cap\sD(\cH)$ , and
\begin{align*}
\sigma'\quad\geq \quad2^{-\alpha}\ \trace(\sigma_{x})\ \rho_{x}
    \quad \geq \quad \frac{2^{-\alpha}}{1+\epsilon^{2}}\ \rho_{x}\enspace.
\end{align*}

This inequality is precisely in the form of inequality~(\ref{eq:RSP protocol}). Now
we run the JRS protocol to approximate~$Q(x)$, with~$t=2^\alpha(1+\epsilon^{2})\ln\frac{2}{\epsilon^{4}}$.
At the end of this protocol, Bob's output is
\[
\tilde{\sigma}_{x}\quad \coloneqq \quad\left( 1- \left(1-2^{-\kappa}
\right)^{t}\right)\sigma_{x}+(1-2^{-\kappa})^{t}\frac{\id}{\dim(\cH)}\enspace,
\]
where~$\kappa=\alpha+\log(1+\epsilon^{2})$. 

By joint concavity of fidelity, and because~$\sigma_{x}$ is 
$\frac{\epsilon}{\sqrt{1+\epsilon^{2}}}$-close to~$Q(x)$, we have
\begin{align*}
\rF(Q(x),\tilde{\sigma}_{x})\quad  \geq\quad\big(1-(1-2^{-\kappa})^{t}\big) \  \rF(Q(x),\sigma_{x})
    \quad \geq \quad\frac{1-(1-2^{-\kappa})^{t}}{\sqrt{1+\epsilon^{2}}}
    \quad \geq \quad \sqrt{1 - \epsilon^2} \enspace.
\end{align*}
Here we appealed to the inequalities~$\ln(1-x)\leq-x$ and~$\sqrt{1-x}
\leq 1 - \tfrac{x}{2}$ (for~$x \in [0,1)$), and the definition of~$\kappa$
and~$t$.
Thus, the purified distance of~$Q(x)$ and~$\tilde{\sigma}_x$ is at most~$\epsilon$,
and the protocol performs remote state preparation with worst-case
error~$\epsilon$. The communication cost of this protocol is~$\lceil \log(t+1)\rceil$. Hence, we have
\begin{align*}
\sQ^*(\RSP(S,Q),\epsilon) \quad \leq \quad \lceil \log(t+1)\rceil
\quad \leq  \quad \alpha+\log_2(1+\epsilon^2)+\log_2\ln \frac{2}{\epsilon^4}+2 \enspace,
\end{align*} 
the stated upper bound.
\end{proof}

\subsection{A lower bound}

By definition, any protocol with worst-case error at most~$\epsilon$
is also a protocol with average-case error at most~$\epsilon$. As a consequence, any lower bound for average-case communication complexity is also a lower bound for worst-case communication complexity.
In particular, by Theorem~\ref{thm-lb}, 
for each probability distribution~$p$, $\Imax^{\epsilon}(A:B)_{\rho(p)}$
is a lower bound for the worst-case communication complexity of remote
state preparation. Therefore,
\begin{equation}
\max_{p}\:\Imax^{\epsilon}(A:B)_{\rho(p)}\quad\leq\quad\sQ^{*}(\RSP(S,Q),\epsilon)\enspace,\label{eq:worst-case lb(1)}
\end{equation}
where the maximum is over all probability distributions~$p$ on the set~$S$.
In the following theorem, we give a lower bound for~$\sQ^*(\RSP(S,Q),\epsilon)$ in terms of max-relative entropy using Equation~\eqref{eq:worst-case lb(1)}.

\begin{theorem}\label{thm-worst-case lb} Let~$S$ be a non-empty finite set,
$Q:S\rightarrow\sD(\cH)$ be a function from~$S$ to the set of density
operators in Hilbert space~$\cH$,~$\epsilon\in(0,1]$, and~$\delta \in
(0,1-\epsilon^2)$. Then
\[
\min_{\sigma\in\sD(\cH)}\max_{x\in S}\ \Dmax^{\gamma}(Q(x)\|\sigma)-\log\frac{(1-\epsilon^{2})(\epsilon^{2}+\delta)}{\delta^{3}}-3\log3\quad\leq\quad\sQ^{*}(\RSP(S,Q),\epsilon)\enspace,
\]
where~$\gamma=\sqrt{2(\epsilon^{2}+\delta)}$.

\end{theorem}
\begin{proof}
By definition of the smooth max-information, Eq.~\eqref{eq:worst-case lb(1)} implies that
\begin{align}
\max_{p}\min_{\sigma\in\sD(\cH)}\ \Dmax^{\epsilon}(\rho_{AB}(p)\ \|\ \rho_{A}(p)\otimes\sigma)\quad\leq\quad\sQ^{*}(\RSP(S,Q),\epsilon)\enspace, \label{eq:first lb for worst(Dmax)} 
\end{align}
whereas the upper bound shown in Theorem~\ref{thm-worst-upperbound} is
\[ \min_{\sigma\in\sD(\cH)} \max_{x\in S} \; \Dmax^\delta (Q(x)\|\sigma)\enspace.
\]
If the minimax theorem held for the above expression, the theorem would
follow. 
However, smooth max-relative entropy~$\Dmax^\epsilon$ is neither convex
nor concave in its arguments, and the minimax theorem does not apply
directly.
Instead, we appeal to Theorem~\ref{thm-Dmax-Dh}, and
approximate it with
hypothesis testing relative entropy~$\Dh^\epsilon$, and write it in
terms of the hypothesis
testing error~$\upbeta^\epsilon$. This measure satisfies the hypotheses of the minimax theorem
(cf.\ Proposition~\ref{thm-concavity of Be in sigma} and~\ref{thm-covexity of Be in p}). We then apply the minimax theorem, and finally return to~$\Dmax^\epsilon$
to derive the lower bound.

By Theorem~\ref{thm-Dmax-Dh}, we have
\begin{align*}
\max_{p}\min_{\sigma\in\sD(\cH)}\ \Dmax^{\epsilon}(\rho_{AB}(p)\ \|\ \rho_{A}(p)\otimes\sigma)\quad\geq & \quad\max_{p}\min_{\sigma\in\sD(\cH)}\ \Dh^{\lambda}(\rho_{AB}(p)\ \|\ \rho_{A}(p)\otimes\sigma)-f(\epsilon,\delta)\\
= & \quad\max_{p}\min_{\sigma\in\sD(\cH)}\ \left(-\log\ \upbeta^{\lambda}(\rho_{AB}(p)\ \|\ \rho_{A}(p)\otimes\sigma)\right)\\
&\quad \mbox{} + \log(1-\lambda)-f(\epsilon,\delta)\\
= & \quad-\log\left(\min_{p}\max_{\sigma\in\sD(\cH)}\ \upbeta^{\lambda}(\rho_{AB}(p)\ \|\ \rho_{A}(p)\otimes\sigma)\right)\\
&\quad \mbox{} + \log(1-\lambda)-f(\epsilon,\delta)\enspace,
\end{align*}
where~$f(\epsilon,\delta)=\log\frac{(1-\epsilon^{2})(\epsilon^{2}+\delta)}{\delta^{3}}+3\log3$
and~$\lambda=1-\epsilon^{2}-\delta$. 

Let~$A_{1}$ be the set of all probability distributions~$p$ over~$S$, and
$A_{2}$ be the set of all quantum states~$\sigma\in\sD(\cH)$. Viewing
$\sigma$ as an element of the real vector space of Hermitian operators
in~$\sL(\cH)$,~$A_{1}$ and~$A_{2}$ are non-empty, convex and compact
subsets of~$\reals^{n}$ for some positive integer~$n$. The quantity
$\upbeta^{\lambda}(\rho_{AB}(p) \,\|\,
\rho_{A}(p)\otimes\sigma)$ is a continuous function of its arguments. 
Moreover, by Proposition~\ref{thm-covexity of Be in p} and 
Proposition~\ref{thm-concavity of Be in sigma}, it
satisfies both conditions of the minimax theorem, Theorem~\ref{thm-minimax theorem}. Thus, we conclude that
\begin{align}
\nonumber
\max_{p}\min_{\sigma\in\sD(\cH)}\ \Dmax^{\epsilon}(\rho_{AB}(p)\ \|\
\rho_{A}(p)\otimes\sigma)\quad\geq &
\quad-\log\left(\max_{\sigma\in\sD(\cH)}\min_{p}\ \upbeta^{\lambda}(\rho_{AB}(p)\ \|\ \rho_{A}(p)\otimes\sigma)\right)\\
&\quad \mbox{} + \log(1-\lambda)-f(\epsilon,\delta)\nonumber\\
\nonumber
= & \quad\min_{\sigma\in\sD(\cH)}\max_{p}\ \Dh^{\lambda}(\rho_{AB}(p)\ \|\ \rho_{A}(p)\otimes\sigma)-f(\epsilon,\delta)\\
\nonumber
\geq & \quad\min_{\sigma\in\sD(\cH)}\max_{p}\ \Dmax^{\gamma}(\rho_{AB}(p)\ \|\ \rho_{A}(p)\otimes\sigma)-f(\epsilon,\delta)\\
\geq & \quad\min_{\sigma\in\sD(\cH)}\max_{x\in S}\ \Dmax^{\gamma}(Q(x)\|\sigma)-f(\epsilon,\delta)\enspace,\label{eq:minimax applied} 
\end{align}
where~$\gamma=\sqrt{2(1-\lambda)}=\sqrt{2(\epsilon^{2}+\delta)}$.
In the second inequality above, we use Theorem~\ref{thm-Dmax-Dh} to move between hypothesis testing relative entropy and max-relative entropy. Combining Eqs.~\eqref{eq:minimax applied} and~\eqref{eq:first lb for worst(Dmax)}, we get the lower bound for the worst-case communication complexity of ARSP.
\end{proof}

\section{Some observations\label{sec:Some-Observations}}

In earlier sections, we characterized the communication complexity
of the approximate remote state preparation problem (ARSP) for both worst-case
error and average-case error. We now discuss the results, especially in light of previous work.

\subsection{A comparison with previous works }

In Section~\ref{sec:Worst-case-error-communication}, we derived 
bounds on the worst-case communication complexity
of ARSP. Jain~\cite{jain_communication_2006} showed that
the worst-case communication complexity of ARSP of a sequence of quantum
states~$(Q(x) : x \in S)$ is bounded from above in terms of
the ``maximum possible information''~$\sT(Q)$ as:
\begin{align}\label{eq:Jain-upperbound}
\frac{8(4\sT(Q)+7)}{\left(1-\sqrt{1-\epsilon^{2}}\right)^{2}}\enspace, 
\end{align}
where~$\epsilon$ is the approximation error. (See
Section~\ref{sec:Asymptotic Information Theory} for a definition
of~$\sT(Q)$.)

We observe that for certain sets of states 
there is a large separation between the
bound established in Theorem~\ref{thm-worst-upperbound},
and Equation~\eqref{eq:Jain-upperbound}. Specifically, the upper bound in Theorem~\ref{thm-worst-upperbound} may be asymptotically smaller than the bound in Equation~\eqref{eq:Jain-upperbound}.

The separation follows from a combination of two pieces of work.
The first is an information-theoretic result, the \emph{Substate
theorem} due to Jain, Radhakrishnan, and Sen~\cite{jain_property_2009},
which
relates the smooth max-relative entropy of two states to their observational
divergence. The precise form of the statement below is due to Jain and
Nayak~\cite{jain_short_2012}.

\begin{theorem}[Substate
theorem~\cite{jain_property_2009,jain_short_2012}]
\label{thm-substate theorem}
Let~$\cH$ be a Hilbert
space, and let~$\rho,\sigma\in\sD(\cH)$ be quantum states such that
$\support(\rho)\subseteq \support(\sigma)$. For any~$\epsilon\in(0,1)$,
\[
\Dmax^{\epsilon}(\rho\|\sigma)\quad\leq\quad\frac{\rD(\rho\|\sigma)}{\epsilon^{2}}+\log\frac{1}{1-\epsilon^2}\enspace.
\]
\end{theorem}

The second result is due to Jain, Nayak, and
Su~\cite{jain_separation_2010}, who constructed an ensemble of quantum
states for which there is a large separation between its Holevo and
Divergence information. (See Section~\ref{sec:Asymptotic Information
Theory} for a definition of these two information quantities.) 

\begin{theorem} 
\label{thm-separationD&chi}
Let~$n$ be a positive
integer, and~$\cH$ be a Hilbert space of dimension
$n$. For every positive real number~$k\geq1$ such that~$\log_2 n > 36 k^{2}$,
there is a finite set~$S$ and an ensemble~$\scE=\{(\lambda_{x},\xi_{x}):x\in S\}$ of quantum
states~$\xi_{x}\in\sD(\cH)$ with~$\xi \coloneqq \sum_{x\in S}\lambda_{x}\xi_{x}=\frac{\id}{n}$,
such that~$\rD(Q(x)\|\xi)=\rD(\scE)=k$ for all~$x\in S$ and~$\chi(\scE)\in\Theta(k\log\log n)$. 
\end{theorem}
Jain~$\etal$~\cite{jain_separation_2010} also showed that this is the best separation
possible for an ensemble of quantum states with a completely mixed
ensemble average. 

Putting these together, we get:
\begin{theorem} 
\label{thm-separationWorst&Avg}
Let~$\delta \in (0,1]$ and~$\cH$ be Hilbert space with dimension~$n$.
Then, for every positive
real number~$k\geq1$ such that~$\log_2 n > 36k^{2}$, there is a finite
set~$S$ and a function
$Q:S\rightarrow\sD(\cH)$ such that~$\sT(Q)\in\Omega(k\log\log n)$
while
\[
\min_{\sigma\in\sD(\cH)}\max_{x\in S}\ \Dmax^{\delta}(Q(x)\|\sigma)\quad\leq\quad\frac{k}{\delta^{2}}+\log\frac{1}{1-\delta^{2}}\enspace.
\]
\end{theorem}

\begin{proof}
Let~$S$ be the set~$S$ and~$\scE=\{(\lambda_{x},\xi_{x}):x\in S\}$ the
ensemble given by
Theorem~\ref{thm-separationD&chi}. Let~$Q:S\rightarrow\sD(\cH)$
be the function such that~$Q(x)=\xi_{x}$ for all~$x\in S$. Suppose
that~$\xi \coloneqq \sum_{x\in S}\lambda_{x}\xi_{x}$ is the ensemble average.
Then we have
\begin{align*}
\min_{\sigma\in\ \sD(\cH)}\max_{x\in S}\ \Dmax^{\delta}(Q(x)\|\sigma)\quad\leq & \quad\max_{x\in S}\ \Dmax^{\delta}(Q(x)\|\xi)\\
\leq & \quad\frac{\max_{x}\rD(Q(x)\|\xi)}{\delta^{2}}+\log\frac{1}{1-\delta^{2}}\\
= & \quad \frac{k}{\delta^2}+\log \frac{1}{1-\delta^2}\enspace,
\end{align*}
where the second inequality is derived using the Substate theorem
(Theorem \ref{thm-substate theorem}). Moreover, by definition of the
maximum possible information~$\sT(Q)$,
we have~$\sT(Q)\geq\chi(\scE)$. This gives us the existence of the 
required function~$Q$.
\end{proof}

Jain~\cite{jain_communication_2006} also gave a
lower bound of~$\sT(Q)/2$ for exact remote state preparation.
The above observation also implies that allowing remote state
preparation with non-zero error in approximating the state may decrease 
the communication cost asymptotically.
By Theorem~\ref{thm-separationWorst&Avg},
we get a function~$Q$ for which the worst-case 
complexity with zero error~$\sQ^{*}(\RSP(S,Q),0) \in\Omega(k\log\log n)$,
while for any~$\epsilon \in (0,1]$, the complexity with error~$\epsilon$
is
\[
\sQ^{*}(\RSP(S,Q),\epsilon)\quad\leq\quad\frac{k}{\delta^{2}}+\log
\frac{1}{1-\delta^2} \enspace,
\]
where~$\delta \coloneqq \frac{\epsilon}{2\sqrt{1+\epsilon^{2}}}$.

\subsection{Average-case error vs.\ worst-case error}
\label{sec:worst-vs-avg}

Requiring bounded worst-case error in approximating states in remote
state preparation is more
demanding, and potentially requires more communication, as compared to
the average case. Here we quantify how much more expensive it could be.

For the rest of this subsection, we let~$n$ be a positive integer,
fix~$S = \set{1, 2, \dotsc, 2^n}$, $\cH = \Span \set{ \ket{x} : x
\in S}$, and define~$Q : S \rightarrow \sD(\cH)$ by~$Q(x) = \density{x}$
for all~$x \in S$. 

\begin{proposition}
\label{thm-worst-vs-avg}
For every~$\epsilon\in [0,1/\sqrt{2}\,)$, there is a probability 
distribution~$p_\epsilon$
over the set~$S$ such that $\sQ_{p_\epsilon}^{*}(\RSP(S,Q),\epsilon) = 0$,
while~$\sQ^{*}(\RSP(S,Q),\epsilon) \ge n$.
\end{proposition}
Using quantum teleportation, any set of quantum states in 
space~$\cH$ can be prepared with zero error with communication 
cost~$2 n$. Thus, the above separation is maximal, up to the factor
of~$2$.

To prove Proposition~\ref{thm-worst-vs-avg}, we first analyze worst-case 
error protocols.
\begin{lemma}
\label{lem-worst-case}
For any~$\epsilon\in [0, 1/\sqrt{2} \,)$,
$\sQ^{*}(\RSP(S,Q),\epsilon)\quad\geq\quad n $.
\end{lemma}

\begin{proof}
Given any ARSP protocol~$\Pi$ for the given set of states~$Q$, we 
construct an LOCC protocol~$\Pi'$ for transmitting~$n$ bits:

\textbf{Protocol~$\Pi'$}
\begin{enumerate}
\item Alice, with input~$x \in S$, and Bob (with no input) simulate the
protocol~$\Pi$.

\item
Let~$\sigma_x$ be the output of~$\Pi$, obtained by Bob.
Bob measures~$\sigma_x$ according to the projective 
measurement~$(\density{y} : y \in S)$.
\end{enumerate}
The communication complexity of~$\Pi'$ equals that of~$\Pi$.

Suppose Alice is given a uniformly random input, and 
let~$X$ be the corresponding random variable.
Let~$Y$ be the random variable corresponding to Bob's output in~$\Pi'$.
Then, by the monotonicity of fidelity under quantum channels,
the success probability of~$\Pi'$ is
\begin{align*}
\Pr[Y = X]\quad \ge & \quad\frac{1}{2^{n}}\sum_{x}\rF(\sigma_{x},Q(x))^{2}\quad \geq \quad1-\epsilon^{2}\enspace.
\end{align*}
By Theorem \ref{thm-convey info by LOCC}, 
the communication cost of~$\Pi'$, and therefore of~$\Pi$,
is at least~$n+\log(1-\epsilon^{2})$.
Since~$\epsilon\in[0,\tfrac{1}{\sqrt{2}})$, we have~$\log(1-\epsilon^{2})>-1$.
So~$\sQ^{*}(\RSP(S,Q),\epsilon) \geq n $.
\end{proof}

We show that the complexity of the task drops drastically, if
average-case error is considered.

\begin{lemma}
\label{lem-avg-case}
For every~$\epsilon\in [0,1/\sqrt{2}\,)$,
There is a probability distribution~$p_\epsilon$ over the set~$S$ such
that $\sQ_{p_\epsilon}^{*}(\RSP(S,Q),\epsilon) = 0$.
\end{lemma}
\begin{proof}
Fix some~$x_0 \in S$.
Let~$p_\epsilon$ be the probability distribution defined by
\begin{align*}
p_{\epsilon,x} \quad=\quad
\begin{cases}
\sqrt{1-\epsilon^{2}} & x=x_{0}\\
\frac{1-\sqrt{1-\epsilon^{2}}}{2^{n}-1} & x\neq x_{0} \enspace.
\end{cases}
\end{align*}
Consider the protocol~$\Pi$ in which Alice does not send any message
to Bob, and Bob always prepares the state~$Q(x_{0})=\density{x_{0}}$.
The final joint state of the input-output registers in the 
protocol~$\Pi$ is
$$\rho'_{AB}\quad =\quad
    \sum_{x\in S}p_{\epsilon,x}\density{x}\otimes Q(x_{0})$$
and the communication cost is zero. Denoting by~$\rho_{AB}$ the ideal
input-output state, we have
\begin{align*}
\rF(\rho_{AB},\rho'_{AB})\quad\geq\quad\sqrt{1-\epsilon^{2}}\enspace.
\end{align*}
So~$\sQ_{p}^{*}(\RSP(S,Q),\epsilon)=0$.
\end{proof}

Thus we conclude Proposition~\ref{thm-worst-vs-avg}.
In fact we can construct an ensemble independent of~$\epsilon$, 
which exhibits a similar disparity between worst and average-case ARSP. 
\begin{proposition}
\label{thm-worst-vs-avg2}
There is a probability distribution~$p$ over~$S$ such that
for every~$\epsilon\in [0,1/\sqrt{2}\,)$, we have
\begin{align*}
\sQ^*_{p}(\RSP(S,Q),\epsilon) \quad \le \quad 
    \log \Big( \min \Big\{ 2^n,
    \log_2 \frac{2}{\epsilon^2} \Big\} \Big) + 2\enspace.
\end{align*}
\end{proposition}
\begin{proof}
Let~$m \coloneqq 2^n$.
Define~$p$ as the geometrically decreasing probability distribution
\[
p_{x}\quad=\quad\begin{cases}
\frac{1}{2^{x}} & x\in\{1,\ldots,m-1\}\\
\frac{1}{2^{m-1}} & x=m
\end{cases}\enspace .
\]
Now consider the following protocol~$\Pi$ for ARSP.
If Alice's input~$x$ belongs to the set~$\{1,\ldots,t\}$
with~$t=\min\{\lceil \log \frac{2}{\epsilon^2}\rceil,m\}$, then she
sends~$x$ to Bob. Otherwise, she sends a random number chosen from the
set~$\{1,\ldots,t\}$ to Bob. After receiving Alice's message~$y$, Bob
outputs the state~$Q(y)$.

In protocol~$\Pi$, the final state of Alice and Bob is of the form
$$\rho'_{AB} \quad \coloneqq \quad \sum_{x=1}^m p_x \density{x}\otimes \sigma_x \enspace,$$
where~$\sigma_x = Q(x)$ for~$x \le t$.
Consequently
\begin{align*}
\rF(\rho_{AB},\rho'_{AB})\quad
= &\quad \sum_{x=1}^m p_x \; \rF(Q(x),\sigma_x) \quad \geq \quad
\sum_{x=1}^t p_x \quad \geq \quad \sqrt{1-\epsilon^2}\enspace.
\end{align*}
Therefore, the average-case error is at most~$\epsilon$, and the communication is~$\lceil \log t\rceil$. This implies that
$$\sQ^*_{p}(\RSP(S,Q),\epsilon)\quad \leq \quad \log \Big( \min\Big\{ 2^n,
\log \frac{2}{\epsilon^2} \Big\}\Big)+2\enspace,$$
as claimed.
\end{proof}

This example illustrates that the more sharply skewed the probability
distribution over~$Q$, the bigger the gap between the worst-case and the
average-case is. The example in Lemma~\ref{lem-avg-case} is a limiting 
case of such a distribution.

\subsection{Connection to the asymptotic case}

It is worth mentioning that our bounds for the
average-case communication complexity of ARSP in the one-shot scenario
also gives the optimal bounds in the asymptotic scenario established
earlier by Berry
and Sanders~\cite{berry_optimal_2003}. This can be derived using the
\emph{Quantum Asymptotic Equipartition Property} of max-information,
i.e., Theorem~\ref{thm-QAEP-Imax}. In the asymptotic scenario, Alice is
given~$n$ independent and identically distributed inputs. Using the notation from Section~\ref{sec:Average-case-error-communication}, the
target joint state of Alice's input and Bob's output is~$\rho(p)^{\tensor n}$, and the goal is to prepare it approximately on Bob's side with average error~$\epsilon$.

Let~$\sq^*_{p}(\RSP(S,Q),\epsilon)$ denote the asymptotic rate of
communication complexity of ARSP with average error~$\epsilon$. This is
the limit of the communication complexity of preparing~$\rho(p)^{\tensor n}$ with
average-case error~$\epsilon$, divided by~$n$, as~$n\rightarrow \infty$.  By Theorems~\ref{thm-Ub} and~\ref{thm-lb}, we have
$$\lim_{n\rightarrow\infty} \ \frac{1}{n} \;
\Imax^{\epsilon}(A:B)_{\rho(p)^{\tensor n}} \quad \leq \quad
\sq^*_p(\RSP(S,Q),\epsilon) \quad \leq \quad \lim_{n\rightarrow\infty} \ \frac{1}{n} \left( \Imax^{\delta}(A:B)_{\rho(p)^{\tensor n}}+\log_{2} \ln\frac{8}{\epsilon^{2}}+2\right)\enspace,$$
where~$\delta=\frac{\epsilon}{2\sqrt{2}}$. So by
inequalities~\eqref{eq-QAEP-lb} and~\eqref{eq-QAEP-ub} in
Theorem~\ref{thm-QAEP-Imax}, we get the following bounds:
$$\rI(A:B)_{\rho(p)}-2\epsilon \log(\abs{A}\abs{B}) \quad \leq \quad
\sq^*_p(\RSP(S,Q),\epsilon) \quad \leq \quad \rI(A:B)_{\rho(p)}\enspace.$$

\section{On LOCC protocols for transmitting bits}
\label{sec-locc-bits}

In this section, we digress from the main theme of this article; we
characterize the communication required to convey classical bits through
LOCC protocols as in Theorem~\ref{thm-convey info by LOCC}.
We have used this in Section~\ref{sec:Some-Observations} 
to highlight a key difference between
worst-case and average-case protocols for remote state preparation.

Consider the following communication task~$\cT$:
\begin{quote}
Two physically separated parties, Alice and Bob, have unlimited
computational power and can communicate with each other.
Alice is given a uniformly random~$n$-bit string~$X$ unknown
to Bob, that is independent of their initial state.
Alice and Bob communicate with each other so that Bob learns~$X$ with
probability at least~$p \in (0,1]$.
\end{quote}

Consider a classical communication protocol in which Alice sends
exactly~$\lceil n-\log \frac{1}{p}\rceil$ bits of~$X$,
and Bob chooses uniformly random bits as his guess for the remaining 
bits. Then the probability that Bob correctly decodes Alice's message is
at least~$p$. In this section, we show that even if we allow Alice and
Bob to use LOCC protocols, the classical communication complexity of the
task~$\cT$ does not decrease. In other words, in any (potentially two-way)
LOCC protocol for this task, Alice sends at least~$n+\log p$ bits in 
order to achieve success probability at least~$p$
(Theorem~\ref{thm-convey info by LOCC}).
Nayak and Salzman~\cite{nayak_limits_2006} showed that in any two-way
\emph{quantum\/} communication protocol with shared entanglement for the
task~$\cT$, Alice sends at least~$\frac{1}{2}(n+\log p)$ qubits to Bob. 
We obtain Theorem~\ref{thm-convey info by LOCC} by strengthening their proof.

\subsection{Preparation}
\label{Sec-Prep}

In LOCC protocols we assume
that Alice and Bob each have access to an arbitrarily large but finite
supply of qubits in some fixed basis state, say~$\ket{\bar{0}}$. Without
loss of generality, we further assume that during a protocol, each party
performs some unitary operation followed by the
measurement of a subset of qubits in the standard basis. 
Note that any measurement can be implemented in this
manner~\cite[Sec 2.2.8]{nielsen_quantum_2000}. Further, if the subset of
qubits measured is of size~$k$, we may assume that it consists of the
leftmost~$k$ qubits.

We state some properties of protocols and states from
Ref.~\cite{nayak_limits_2006} which are used later in this section.
For completeness we include their proofs here.

\begin{proposition}[\cite{nayak_limits_2006}]
\label{thm-arbitrary-entanglement}
In any communication protocol with prior entanglement and local quantum channels, we may assume that the initial shared quantum state is of the form
\[
(\id_{A}\otimes\Lambda)\sum_{r\in\{0,1\}^{e}}\ket{r}_A\ket{r}_B\enspace,
\]
for some~$\Lambda \coloneqq \sum_{r\in\{0,1\}^{e}}\sqrt{\lambda_{r}}\density{r}$
with~$\lambda_r\geq 0$, $\sum_{r\in\{0,1\}^{e}}\lambda_{r}=1$, and for
some integer~$e\geq 1$.
\end{proposition}
\begin{proof}
Without loss of generality, assume that Alice and Bob hold~$e_A$
and~$e_B$ qubits of the initial state, respectively, where~$e_B \geq
e_A$.
Let~$\ket{\phi}=\sum_{i\in\{0,1\}^{e_A}}\sqrt{\gamma_i}\ket{a_i}_A\ket{b_i}_B$
be a Schmidt decomposition of the initial shared state.

We define a new protocol in which
Alice and Bob start with the shared state~$\ket{\psi} \coloneqq
\sum_{r\in\{0,1\}^{e_B}}\sqrt{\lambda_r}\ket{r}_A\ket{r}_B$,
where~$\lambda_{\bar{0}s}=\gamma_s$ for~$s\in\{0,1\}^{e_A}$ and is zero
otherwise. The state simplifies to
$$\sum_{i\in\{0,1\}^{e_A}}\sqrt{\gamma_i}\ket{\bar{0},i}_A \ket{\bar{0},i}_B\enspace.$$
Using appropriate local unitary operators, Alice and Bob produce
the state~$\ket{\phi}$ (tensored with some fixed pure state), and then
run the original protocol.
\end{proof}

\begin{proposition}[\cite{nayak_limits_2006}]
\label{thm:exchange unitary performation party}
For any linear transformation~$T$ on~$e$ qubits and any orthonormal set
$\{\ket{\phi_{a}}:a\in\{0,1\}^{e}\}$ over $e'\geq e$ qubits,
\begin{equation*}\label{eq:exchange unitary performation}
\sum_{a\in\{0,1\}^e}T\ket{a}\otimes\ket{\phi_a}\quad=\quad\sum_{a\in\{0,1\}^{e}}\ket{a}\otimes\tilde{T}\ket{\phi_{a}}\enspace,
\end{equation*}
where~$\tilde{T}$ is any transformation on~$e'$ qubits such that
for all $a'\in\{0,1\}^{e}$,
$\tilde{T}\ket{\phi_{a'}}= \sum_{a \in
\{0,1\}^{e}} \bra{a'}T\ket{a} \ket{\phi_a}$. If~$T$ is a unitary operation, then we may take~$\tilde{T}$ to be a
unitary operation on~$e'$ qubits.
\end{proposition}

\begin{proof} 
Since the set~$\{\ket{a}:a\in\{0,1\}^e\}$ is an orthonormal basis for the Hilbert space of~$e$ qubits, we have
\begin{align*}
\sum_{a\in\{0,1\}^e}T\ket{a}\ket{\phi_a}\quad = &\quad \sum_a\sum_{a'}\bra{a'}T\ket{a}\ket{a'}\ket{\phi_a}\\
= & \quad \sum_{a'}\ket{a'}\sum_a \bra{a'} T \ket{a}\ket{\phi_a}   \\
= &\quad \sum_{a'} \ket{a'}\tilde{T}\ket{\phi_{a'}}\enspace,
\end{align*}
as claimed.
The second part of the proposition is straightforward.
\end{proof}

We also use this property in the following form in our analysis. The proof is
straightforward, and is omitted.
  
\begin{corollary}
\label{cor-EPR-locc}
For any controlled unitary operation~$T \coloneqq \sum_{z \in
\set{0,1}^m} \density{z} \tensor T_z$ on a classical-quantum register
with~$m$ bits and~$e$ qubits, and collections of orthonormal sets
$\{\ket{\psi_{za}}:a\in\{0,1\}^{e}\}$ over $e'$ qubits with~$e' \ge e$ and~$z \in \set{0,1}^m$,
\begin{equation*}
\label{eqn-EPR-locc}
\sum_{z \in \set{0,1}^m} \sum_{a\in\{0,1\}^e}
T\ket{za} \otimes \ket{z}\ket{\psi_{za}}
    \quad = \quad \sum_{z \in \set{0,1}^m} \sum_{a\in\{0,1\}^{e}}
        \ket{za} \otimes \tilde{T} (\ket{z} \ket{\psi_{za}}) \enspace,
\end{equation*}
where~$\tilde{T} \coloneqq \sum_{z \in \set{0,1}^m} \density{z} 
\tensor \tilde{T}_z$, and~$(\tilde{T}_z)$ is a sequence of unitary 
transformations on~$e'$ qubits such that for all~$z \in \set{0,1}^m$
and~$a'\in\{0,1\}^{e}$, $\tilde{T}_z \ket{\psi_{za'}} = \sum_{a \in \set{0,1}^e} \bra{a'} T_z
\ket{a} \ket{\psi_{za}}$.

\end{corollary}

\subsection{One-way LOCC protocols}
\label{sec:One-way-communication-LOCC}

As a warm-up, we prove the analogue of Theorem~\ref{thm-convey info by LOCC}
for one-way LOCC protocols.

\begin{theorem}
\label{thm:convey info by 1-way LOCC}
Let~$Y$ be Bob's output in any one-way LOCC protocol for task~$\cT$ when
Alice receives uniformly distributed~$n$-bit input~$X$. Let~$p \coloneqq \Pr[Y=X]$ be the probability that Bob gets the output~$X$. Then
\[
m \quad\geq\quad n-\log \frac{1}{p}\enspace,
\]
where $m$ is the number of classical bits Alice sends to Bob in the
protocol.
\end{theorem}

\begin{proof} Using Proposition~\ref{thm-arbitrary-entanglement}, we
assume that the initial shared entangled state is~$\sum_{r\in\{0,1\}^{e}}\ket{r}\Lambda\ket{r}$
for some~$\Lambda \coloneqq \sum_{r\in\{0,1\}^{e}}\sqrt{\lambda_{r}}\density{r}$
with~$\lambda_r\geq 0$ and~$\sum_{r\in\{0,1\}^{e}}\lambda_{r}=1$, and
some~$e\geq 1$. As explained in Section~\ref{Sec-Prep}, first Alice performs a
unitary transformation on her part of the initial state depending on her
input~$X$ and measures the left-most~$m$ qubits in the standard basis.
Let~$U_{x}$ be the unitary operation Alice uses when she is given~$x$ as
input. After the unitary operation~$U_{x}$ is performed, the joint state is
\begin{align*}
(U_{x}\otimes\id)(\id\otimes\Lambda)\sum_{r\in\{0,1\}^{e}}\ket{r}\otimes\ket{r}\quad= & \quad(\id\otimes\Lambda)(U_{x}\otimes\id)\sum_{r\in\{0,1\}^{e}}\ket{r}\otimes\ket{r}\\
= & \quad(\id\otimes\Lambda)(\id\otimes
U_{x}^{\transpose})\sum_{r\in\{0,1\}^{e}}\ket{r}\otimes\ket{r} &
\textrm{(By Proposition~\ref{thm:exchange unitary performation party})}\\
= & \quad\sum_{r\in\{0,1\}^{e}}\ket{r}\Lambda U_{x}^{\transpose}\ket{r}\enspace.
\end{align*}
Then Alice measures the state as described above and sends Bob the outcome of her measurement. Bob's state after this step is
\[
\xi_{x} \quad = \quad \sum_{z\in\{0,1\}^{m}}\density{z}\otimes\Lambda
U_{x}^{\transpose}(\density{z}\tensor\id)\overline{U}_{x}\Lambda^{*}\enspace.
\]
Note that
\begin{equation}\label{eq:Note}
\xi_{x} \quad = \quad(\id\tensor\Lambda)\left(
\sum_{z\in\{0,1\}^{m}}\density{z}\otimes
U_{x}^{\transpose}(\density{z}\tensor\id)\overline{U}_{x}\right)
(\id\tensor\Lambda^{*})\quad \leq \quad (\id \tensor \Lambda
\Lambda^{*})\enspace,
\end{equation}
where the identity operator acts on a~$2^m$ dimensional space.
Finally, Bob performs a projective measurement~$\{P_{y}\}_{y\in\{0,1\}^{n}}$ on
his qubits, and gets as outcome the random variable~$Y$. The success probability~$p$ of the protocol is
\begin{align*}
\Pr[X=Y]\quad= & \quad\sum_{x\in\{0,1\}^{n}}\Pr[X=x]\Pr[Y=x|X=x]\\
= & \quad\sum_{x\in\{0,1\}^{n}}\frac{1}{2^{n}} \; \trace\left(P_{x} \xi_{x}\right)\\
\leq & \quad \frac{1}{2^{n}} \sum_x \trace \big(P_x(\id \otimes
\Lambda\Lambda^*)\big) & \textrm{(By equation~\eqref{eq:Note}})\\
= & \quad \frac{1}{2^{n}} \; \trace (\id \otimes \Lambda\Lambda^*)\\
= & \quad \frac{2^m}{2^n}\enspace.
\end{align*}
We conclude that~$m\geq n+\log p$.
\end{proof}

\subsection{The extension to two-way LOCC protocols}
\label{sec:The-extension-generalLOCC}

We now extend the above result to any two-way LOCC protocol. In
particular we prove Theorem~\ref{thm-convey info by LOCC}, which we restate
here for convenience. 

\begin{theorem} \label{thm-convey info by LOCC-revision} Let~$Y$ be
Bob's output in any two-way LOCC protocol for task~$\cT$ when Alice
receives uniformly distributed~$n$-bit input~$X$. Let~$m_A$ be the total
number of bits Alice sends to Bob, and~$p \coloneqq \Pr[Y=X]$ be the 
probability that Bob produces output~$X$. Then
\[
m_{A}\quad\geq\quad n-\log \frac{1}{p}\enspace.
\]
\end{theorem}

To prove the theorem, we characterise the joint state of Alice and Bob at the end of a bounded round LOCC protocol.

\begin{lemma}
\label{thm-joint state in 2way}
Let~$\Pi$ be a bounded round LOCC protocol. Let~$e$ be the initial number
of qubits with each of Alice and Bob,
$q$ be the total number of bits sent by Alice to Bob,
$q'$ be the total number of bits sent by Bob to Alice,
and~$m$ be the total number of bits exchanged in~$\Pi$ (so~$m = q+q'$).
Then Alice and Bob's joint state at the end of the protocol
(before the measurement for producing the output) can be written as
\[
\sum_{z\in\{0,1\}^{m}}\ \sum_{r,s\in\{0,1\}^{e-q}}
\ket{z,r}\!\bra{z,s}_A\tensor 
\Lambda \ket{\phi_{z,r}}\!\bra{\phi_{z,s}}_B \Lambda^{*} \enspace,
\]
where
\begin{enumerate}
\item $A$ and~$B$ are classical-quantum registers with $m$-bit classical
parts that contain the transcript of the protocol; register~$A$ is with
Alice, and~$B$ with Bob,

\item $\Lambda$ is a linear transformation that maps classical-quantum states with~$m$ bits and~$e$ qubits
to classical-quantum states of the same form, depends only on the initial joint state 
and the unitary transformations applied by Bob, and
satisfies~$\trace(\Lambda \Lambda^{*})=2^{q}$;
and

\item $\{\ket{\phi_{z,r}}\}$ is an orthonormal set of classical-quantum 
states of the form~$\ket{\phi_{z,r}} \coloneqq \ket{z}\ket{\psi_{z,r}}$ 
over~$m$-bits and~$e$ qubits, and depends only on the
initial joint state and the unitary transformations applied by Alice.
\end{enumerate}
\end{lemma}

\begin{proof} Suppose that~$\Pi$ is a~$t$-round LOCC protocol.
Let~$\rho_{i}$ be the joint state of Alice and Bob after~$i$-th round,
and~$m_{i}$ be the total number of bits exchanged by Alice and Bob in
the first~$i$ rounds, of which $q_i$ bits are sent by Alice,
for~$1\leq i \leq t$.  Let~$\rho_0$ be their initial state.

We prove the lemma by induction on~$t$.

\textbf{Base Case: }
Suppose that~$\Pi$ is a zero communication LOCC protocol, i.e.,~$t=0$.
By Proposition~\ref{thm-arbitrary-entanglement}, we have
\[
\rho_{0}\quad = \quad \sum_{r,s\in\{0,1\}^{e}}\ketbra{r}{s}\tensor\Lambda\ketbra{r}{s}\Lambda^{*}\enspace,
\]
where~$\Lambda=\sum_{r\in\{0,1\}^{e}}\sqrt{\lambda_{r}}\density{r}$
for some~$\lambda_r\geq 0$ and~$\sum_{r}\lambda_{r}=1$.
Since~$\trace(\Lambda\Lambda^{*})=1$, the state~$\rho_{0}$ satisfies 
the claimed properties.

\textbf{Induction Hypothesis:} Suppose the lemma holds for any~$l$-round LOCC
protocol, for some~$l \ge 0$.

\textbf{Inductive Step: } Suppose that~$\Pi$ is an~($l+1$)-round
protocol. By the induction hypothesis, after the first~$l$ rounds 
of communication we have
\[\rho_{l} \quad = \quad 
    \sum_{z\in\{0,1\}^{m_{l}}} \: \sum_{r,s\in\{0,1\}^{e-q_{l}}}
        \ket{z,r}\!\bra{z,s}\tensor 
        \Lambda_{l}\ket{\phi_{z,r}}\!\bra{\phi_{z,s}}\Lambda_{l}^{*}\enspace,
\]
where $\Lambda_{l}$ and $\ket{\phi_{z,r}}$ satisfy
the properties stated in the lemma. In particular,
suppose~$\ket{\phi_{z,r}} \coloneqq \ket{z}\ket{\psi_{z,r}}$ for each~$z,r$.
We show that at the end of the protocol~$\rho_{l+1}$ is in the required
form as well. Consider the~($l+1$)-th round of~$\Pi$.

\textbf{Case (1): }
Suppose that the communication in the last round is from Alice to Bob.
Alice applies a unitary transformation~$U \coloneqq \sum_z \density{z}
\tensor U_z$, which acts on the quantum
part of her register, controlled by the classical part of her register.
She then measures the~$k$ leftmost qubits in the standard basis, appends
the outcome to the message transcript in her classical register, and 
sends the outcome~$a$ of her measurement 
to Bob. The joint state after applying~$U$ is
\begin{align*}
\MoveEqLeft{
(U\otimes\id)(\id\otimes\Lambda_{l})\left[
\sum_{\substack{r,s\in\{0,1\}^{e-q_{l}} \\
z\in\{0,1\}^{m_{l}}}}\ket{z,r}\!\bra{z,s} 
\tensor \ket{\phi_{z,r}}\!\bra{\phi_{z,s}}\right]
(\id\otimes\Lambda_{l}^{*})(U^{*}\otimes\id) } \\
 & = \quad (\id\otimes\Lambda_{l})
     \left[ \sum_{r,s,z}U\ket{z,r}\!\bra{z,s}U^{*} 
         \tensor \ket{\phi_{z,r}}\!\bra{\phi_{z,s}}
     \right] (\id\otimes\Lambda_{l}^{*}) \\
 & = \quad (\id\otimes\Lambda_{l})
     \left[ \sum_{r,s,z}\ket{z,r}\!\bra{z,s} 
         \tensor \tilde{U} \ket{\phi_{z,r}}\!\bra{\phi_{z,s}}
         \tilde{U}^{*}
     \right] 
     (\id\otimes\Lambda_{l}^{*}) \enspace, 
\end{align*}
where~$\tilde{U} := \sum_z \density{z} \tensor \tilde{U}_z$ is the 
unitary operation given by Corollary~\ref{cor-EPR-locc}. 
After Alice performs her measurement and sends the measurement 
outcome~$a$ to Bob, say he stores the message in register~$M$.
Denote by~$\id_M \tensor \Lambda_{l} \tilde{U}$ the
operator~$\Lambda_{l} \tilde{U}$ on the registers originally
with Bob, extended to include the register~$M$. (The order of the
operators in tensor product does not represent the order of the
registers.) The joint state then may be expressed as below.

\begin{align*}
\rho_{l+1}\quad =\quad \sum_{\substack{r'\!,s'\in\{0,1\}^{e-(q_{l}+k)}
\\ a \in\{0,1\}^{k} \\  z\in\{0,1\}^{m_{l}}}} \ket{za,r'}\!\bra{za,s'}
\otimes (\id_M \tensor \Lambda_{l} \tilde{U}) (\density{z} \tensor
\density{a}_M \tensor \ket{\psi_{z, ar'}}\!\bra{\psi_{z, as'}})
(\id_M \tensor \tilde{U}^{*} \Lambda_{l}^{*})
\enspace,
\end{align*}

where~$\Lambda_{l} \tilde{U}$ acts on the classical-quantum register
with Bob before the message was sent.
We define~$\Lambda_{l+1} \coloneqq \id_M \tensor \Lambda_{l}$,
and~$\ket{\phi_{z',r'}} \coloneqq \ket{za} \tensor \tilde{U}_z
\ket{\psi_{z, ar'}}$, where~$z' \coloneqq za$.
Noting that~$m_{l+1}=m_{l}+k$ and $q_{l+1}=q_{l}+k$, 
we have
\begin{align*}
\rho_{l+1} \quad = \quad
    \sum_{\substack{r',s'\in\{0,1\}^{e-q_{l+1}}\\
        z'\in\{0,1\}^{m_{l+1}}}}
        \ket{z',r'}\!\bra{z',s'}
        \tensor \Lambda_{l+1} \ket{\phi_{z',r'}}\!\bra{\phi_{z',s'}}
            \Lambda_{l+1}^{*}\enspace.
\end{align*}
Further note that~$\trace(\Lambda_{l+1}\Lambda_{l+1}^{*}) = 2^{q_{l+1}}$
and~$\{\ket{\phi_{z',r'}}\}$ is an orthonormal set of the claimed form.

\textbf{Case (2):}
Suppose that the communication in the last round is from Bob to Alice.
Bob applies a unitary transformation~$V \coloneqq \sum_z \density{z}
\tensor V_z$ to the quantum part of his
register, controlled by the classical part of his register.
Then he measures the~$k$ leftmost qubits (say in sub-register~$L$) 
in the standard basis, and appends the outcome~$b$ to the message
transcript, in classical register~$M$.
Finally, he sends the outcome~$b$ of the measurement to Alice.
Denote by~$\id_M \tensor (\bra{b}_L \tensor \id) V \Lambda_{l}$, the
extension of the operator~$(\bra{b}_L \tensor \id) V \Lambda_{l}$ to
include the register~$M$. (Here, the order of the operators in tensor
product does not represent the order of the registers on which they
act. The same applies to the operator~$\Lambda_{l+1}$ defined below.)
The joint state then is as follows.
\begin{align*}
\rho_{l+1} \quad = \quad
    \sum_{\substack{r,s\in\{0,1\}^{e-q_{l}}\\
              b\in\{0,1\}^{k}\\
              z\in\{0,1\}^{m_{l}}}
    } \ket{zb,r}\!\bra{zb,s} \tensor
    (\id_M \tensor (\bra{b}_L \tensor \id) V \Lambda_{l})
    (\density{zb} \tensor \ket{\psi_{z,r}}\!\bra{\psi_{z,s}}) 
    (\id_M \tensor \Lambda_{l}^{*} V^{*} (\ket{b}_L \otimes \id))
    \enspace.
\end{align*}
Note that~$q_{l+1}=q_{l}$, and~$m_{l+1}=m_{l}+k$.
Define~$\Lambda_{l+1} \coloneqq \sum_b \density{b}_M \tensor
(\bra{b}_L \otimes \id ) V \Lambda_{l}$
and~$\ket{\phi_{z',r'}} = \ket{zb} \tensor \ket{\psi_{z,br'}}$,
where~$z' \coloneqq zb$.
It is straightforward to verify that~$\trace{(\Lambda_{l+1}
\Lambda_{l+1}^{*})} = 2^{q_{l+1}}$, the set~$\{\ket{\phi_{z',r'}}\}$ is
of the claimed form, and 
\begin{align*}
\rho_{l+1} \quad = \quad
    \sum_{\substack{r',s'\in\{0,1\}^{e-q_{l+1}}\\z'\in\{0,1\}^{m_{l+1}}}}
        \ket{z',r'}\!\bra{z', s'}
        \tensor \Lambda_{l+1} \ket{\phi_{z',r'}}\!\bra{\phi_{z',s'}}
        \Lambda_{l+1}^{*}\enspace.
\end{align*}
This completes the proof.
\end{proof}

We are ready to prove
Theorem~\ref{thm-convey info by LOCC}, restated in this section as
Theorem~\ref{thm-convey info by LOCC-revision}.

\begin{proofof}{Theorem~\ref{thm-convey info by LOCC-revision}} 
By Lemma~\ref{thm-joint state in 2way}, at the end of any two-way
LOCC protocol, when Alice has input~$x \in \set{0,1}^n$,
Bob's state before performing his final measurement to get~$Y$ is
\[
\xi_x\quad =\quad \sum_{\substack{r\in\{0,1\}^{e-m_{\!A}}\\
z\in\{0,1\}^{m}}} \Lambda\ket{\phi_{z,r}(x)}\!\bra{\phi_{z,r}(x)}\Lambda^{*}\enspace,
\]
for some linear transformation~$\Lambda$
with~$\trace(\Lambda\Lambda^{*})=2^{m_{A}}$ and
 orthonormal set~$\{\ket{\phi_{z,r}(x)}\}_{z,r}$.
The transformation~$\Lambda$ only depends on Bob's unitary
operations and the initial state, and is therefore independent of 
Alice's input~$x$.
Note that
 \begin{equation}
\label{eq:1}
 \xi_x\quad \leq \quad \Lambda \Lambda^{*}\enspace.
 \end{equation}
After Bob performs his final projective
measurement~$\{P_{y}\}_{y\in\{0,1\}^n}$ and gets the output~$Y$, the
probability of correctly recovering an input~$X$ chosen uniformly at
random is
\begin{align*}
p\quad \coloneqq \quad \Pr[Y=X]\quad
= & \quad \frac{1}{2^n}\sum_{x\in\{0,1\}^{n}} \trace (P_x \xi_x)\\
\leq & \quad \frac{1}{2^n}\sum_x \trace (P_x \Lambda \Lambda^*)& \mathrm{(Equation\ \eqref{eq:1})}\\
= & \quad  \frac{1}{2^n} \; \trace (\Lambda \Lambda^*)\quad
= \quad \frac{2^{m_A}}{2^n}\enspace.
\end{align*}
 Therefore, we have~$m_{A}\geq n-\log\frac{1}{p}$, as required.
\end{proofof}

\section{Conclusion\label{part:Conclusions-and-Outlook}}

In this article, we studied the communication complexity of remote state
preparation in the one-shot scenario. Our main results can be summarized as follows:
\begin{itemize}
\item The communication complexity of remote state preparation with
bounded average-case error~$\epsilon$ can be characterized tightly in terms of the smooth max-information Bob's output has about Alice's input.
\item The communication complexity of remote state preparation with
bounded worst-case error~$\epsilon$ can be characterized in terms of a
similar natural expression involving smooth max-relative entropy.
\end{itemize}
The bounds we derive for the worst-case communication complexity are
provably tighter than earlier ones. We also show out how protocols 
that guarantee low worst-case error necessarily use more communication
than those that require low error on average.
In the process, we strengthen a lower bound on LOCC protocols for
transmitting classical bits.
 
In this work, we focused on the remote preparation of a possibly mixed 
quantum state. However, often the quantum state to be remotely prepared  is
entangled with other systems (``the environment''). We can consider the
problem of preparing an approximation of the quantum state such that its
entanglement with other systems does not change significantly. This
problem has been studied in asymptotic
scenario~\cite{bennett_remote_2005,berry_optimal_2003}.
Berta~\cite{berta_single-shot_2008} implicitly studied this problem in
the one-shot scenario by considering the \textit{quantum state merging\/} problem, and showed that the minimal entanglement cost needed for this problem is equal to minus the~$\epsilon$-smooth conditional min-entropy of Alice's register conditioned on the environment, while classical communication is allowed for free. Note that the entanglement cost is defined as the difference between the number of bits of pure entanglement at the beginning and at the end of the process.
It would be interesting to characterize the minimum classical
communication of such ``faithful'' ARSP in terms of non-asymptotic
information theoretic quantities. 

\bibliographystyle{plain}
\bibliography{Bibliography}

\appendix 

\section{Some properties of entropic quantities}
\label{App-proof of properties}

In this section, we present the proofs of some properties of
information-theoretic quantities stated in Section~\ref{sec:Asymptotic
Information Theory}.
For convenience, we restate the properties here.

\begin{proposition}[Proposition~\ref{thm-classical-smi}]
Let~$\rho_{AB}\in\sD(\cH'\tensor\cH)$
be a bipartite quantum state that is classical on~$A$. For any~$\epsilon\ge0$, there exists~$\rho'_{AB}\in\sB^{\epsilon}(\rho_{AB})\cap\sD(\cH'\tensor\cH)$
classical on~$A$ such that 
\[
\Imax^{\epsilon}(A:B)_{\rho}\quad=\quad\Imax(A:B)_{\rho'}\enspace.
\]
\end{proposition}

\begin{proof} Let~$\lambda=\Imax^{\epsilon}(A:B)_{\rho}$, and ~$\tilde{\rho}_{AB}\in\sB^{\epsilon}(\rho_{AB})$
and~$\sigma_{B}\in\sD(\cH)$ be two quantum states for which 
\[
\tilde{\rho}_{AB}\quad\le\quad2^{\lambda} \; \tilde{\rho}_{A}\tensor\sigma_{B}\enspace.
\]
Without loss of generality, we assume that~$\tilde{\rho}_{AB}$ has trace
equal to one, i.e.,~$\tilde{\rho}_{AB} \in
\sB^{\epsilon}(\rho_{AB})\cap\sD(\cH'\otimes\cH)$. If not, we consider
the state~$\omega_{AB} \coloneqq  \tfrac{\tilde{\rho}_{AB}}{\trace
(\tilde{\rho}_{AB})}$ instead of~$\tilde{\rho}_{AB}$. Since~$\rho$
has trace~$1$, $\rP(\omega, \rho) \le \rP(\tilde{\rho}, \rho)$. Further,
$\omega_{AB} \leq 2^\lambda \, \omega_A \tensor \sigma_B$.

Let~$\Phi_{A}:\rL(\cH)\rightarrow\rL(\cH)$ be a quantum-to-classical channel such that:
$$\Phi_{A}(X)=\sum_{i}\bra{e_{i}}X\ket{e_{i}}\density{e_{i}}$$
for all~$X\in\rL(\cH)$, where~$\{\ket{e_{i}}\}$ is the standard
basis for~$\rL(\cH)$. Let~$\rho'_{AB}=(\Phi_{A}\tensor\id_{B})(\tilde{\rho}_{AB})$.
By the definition of~$\rho'_{AB}$ and the monotonicity of purified
distance~$\rho'_{AB}\in\sB^{\epsilon}(\rho_{AB})\cap \sD(\cH'\otimes\cH)$.

By optimality of~$\tilde{\rho}_{AB}$, we have
\[
\Imax^{\epsilon}(A:B)_{\rho}\quad=\quad\Imax(A:B)_{\tilde{\rho}}\quad\leq\quad \Imax(A:B)_{\rho'} \enspace,
\]
and by Proposition~$\ref{thm-monotonicity-mi}$, monotonicity of smooth
max-information, we have
\[
\Imax(A:B)_{\rho'}\quad\leq\quad\Imax(A:B)_{\tilde{\rho}}\enspace.
\]
Therefore, we conclude that
\[
\Imax^{\epsilon}(A:B)_{\rho}\quad=\quad\Imax(A:B)_{\rho'}\enspace,
\]
where~$\rho'_{AB}\in\sB^{\epsilon}(\rho_{AB})\cap \sD(\cH'\otimes\cH)$ and is classical on~$A$.
\end{proof}

\begin{proposition} [Proposition~\ref{thm-covexity of Be in p}]
Let~$\rho_{AB}(p)\in\sD(\cH_{A}\otimes\cH_{B})$ be a state classical
on~$A$ such that the distribution on~$A$ is given by the probability
vector $p$. Let~$\rho_{A}(p)=\trace_{B}(\rho_{AB}(p))$,
and~$\sigma\in\sD(\cH_{B})$ be a quantum state on Hilbert
space $\cH_{B}$. Then the function $\upbeta^{\epsilon}(\rho_{AB}(p)
\,\|\, \rho_{A}(p)\otimes\sigma)$
is convex with respect to~$p$.
\end{proposition}

\begin{proof}
Let~$p_{0}$ and~$p_{1}$ be two
arbitrary probability distributions on the standard basis of~$\cH_A$.
For~$\lambda\in[0,1]$, let~$q=\lambda p_{0}+(1-\lambda)p_{1}$.
We show that 
\begin{align*}
\upbeta^{\epsilon}(\rho_{AB}(q) \,\|\, \rho_{A}(q)\otimes\sigma)
    \quad\leq\quad \lambda \; \upbeta^{\epsilon}(\rho_{AB}(p_{0})
        \,\|\, \rho_{A}(p_{0}) \otimes\sigma) + (1-\lambda) \;
    \upbeta^{\epsilon}(\rho_{AB}(p_{1}) \,\|\,
        \rho_{A}(p_{1})\otimes\sigma) \enspace,
\end{align*}
which proves the claim.

Let $\Phi:\sL(\cH_{A})\rightarrow\sL(\complex^{2}\otimes\cH_{A})$
be the quantum channel with Kraus operators $A_{a,x}=\sqrt{\alpha_{x}^{a}}\ket{x}\otimes\density{a}$
for all $a$ and $x\in\{0,1\}$, where $\alpha_{0}^{a}\coloneqq\lambda\frac{p_{0}(a)}{q(a)}$
and $\alpha_{1}^{a}=(1-\lambda)\frac{p_{1}(a)}{q(a)}$. Then we have 
\[
\rho_{X\! AB}(q) \quad = \quad (\Phi\otimes\id_{B})(\rho_{AB}(q))
    \quad = \quad \lambda \, \density{0} \tensor \rho_{AB}(p_0)
        + (1 - \lambda) \, \density{1} \tensor \rho_{AB}(p_1) \enspace.
\]
Since~$\rho_{X\! AB}(q)$ is an extension of~$\rho_{AB}(q)$, 
using Proposition~\ref{thm-DPI-beta} twice, we get
\begin{align}
\label{eqn-b2}
\upbeta^{\epsilon}(\rho_{AB}(q) \,\|\, \rho_{A}(q)\otimes\sigma)
    \quad=\quad \upbeta^{\epsilon}(\rho_{X\! AB}(q) \,\|\, \rho_{X\! A}(q)
        \otimes \sigma) \enspace.
\end{align}

For each $x\in\{0,1\}$, let $Q_{x}$ be the measurement operator 
that achieves~$\upbeta^{\epsilon}(\rho_{AB}(p_{x})
\,\|\, \rho_{A}(p_{x})\otimes\sigma)$. 
Consider the measurement operator $Q \coloneqq 
\sum_{x \in \set{0,1}} \density{x} \otimes Q_{x}$. This satisfies
\begin{align*}
\innerproduct{Q}{\rho_{X\! AB}(q)} 
    \quad = & \quad \lambda \, \innerproduct{Q_0}{\rho_{AB}(p_{0}) }
        + (1-\lambda) \, \innerproduct{Q_1}{\rho_{AB}(p_{1})} 
    \quad \geq \quad 1-\epsilon \enspace,
\end{align*}
by definition of~$Q_{0}, Q_{1}$. By Eq.~(\ref{eqn-b2}) and the
definition of~$\upbeta^\epsilon$, we get
\begin{align*}
\upbeta^{\epsilon}(\rho_{AB}(q) \,\|\, \rho_{A}(q)\otimes\sigma)
    \quad = & \quad \upbeta^{\epsilon}(\rho_{X\! AB}(q) \,\|\,
        \rho_{X\! A}(q)\otimes\sigma) \\
    \quad \leq & \quad \innerproduct{Q}{\rho_{X\! A}(q) \tensor \sigma}
        \\
    = & \quad \lambda \, \innerproduct{Q_0}{\rho_{A}(p_{0}) \tensor
        \sigma} + (1-\lambda) \, \innerproduct{Q_1}{\rho_{A}(p_{1})
        \tensor \sigma} \\
    = & \quad \lambda \, \upbeta^{\epsilon}(\rho_{AB}(p_{0}) \,\|\,
        \rho_{A}(p_{0})\otimes\sigma) + (1-\lambda) \,
        \upbeta^{\epsilon}(\rho_{AB}(p_{1}) \,\|\,
        \rho_{A}(p_{1})\otimes\sigma) \enspace,
\end{align*}
as we set out to prove.
\end{proof}

\begin{proposition}[Proposition~\ref{thm-concavity of Be in sigma}]
For any fixed quantum state~$\rho\in\sD(\cH)$, the function
$\upbeta^{\epsilon}(\rho\|\sigma)$ is a concave function with 
respect to $\sigma$.
\end{proposition}

\begin{proof}
For any choice of~$\sigma_0,\sigma_1 \in \sD(\cH)$
and~$\lambda\in[0,1]$, let~$Q$ be the measurement operator that 
achieves hypothesis testing error
$\upbeta^\epsilon(\rho \,\|\, \lambda\sigma_0+(1-\lambda)\sigma_1)$. Then
\begin{align*}
\upbeta^\epsilon(\rho \,\|\, \lambda\sigma_0+(1-\lambda)\sigma_1)\quad =
&\quad \langle Q, ~ \lambda\sigma_0+(1-\lambda)\sigma_1 \rangle\\
= &\quad \lambda \, \langle Q,\sigma_0 \rangle + (1-\lambda)\,  \langle
Q,\sigma_1 \rangle\\
\geq &\quad \lambda \, \upbeta^\epsilon(\rho\|\sigma_0)+(1-\lambda) \,
\upbeta^\epsilon(\rho\|\sigma_1)\enspace,
\end{align*}
since~$\langle Q,\rho \rangle \geq 1-\epsilon$.
\end{proof}

\section{Preparing states from an infinite set}
\label{sec-RSP-infinite-set}

In this section, we discuss remote state preparation of states
drawn from an infinite set. This scenario has been studied 
by Lo~\cite{lo_classical-communication_2000} and in later works on the
topic.

In remote state preparation, Alice's input is supposed to provide a complete 
description of
the state to be prepared at Bob's end. In any physically realistic model
of computation, the description necessarily has finite bit-length
(see, e.g., Ref.~\cite{Aaronson05}). For instance, if a~$d$-dimensional 
quantum state is described by specifying~$\Theta(d^2)$ complex entries in the
corresponding~$d \times d$ matrix, the complex numbers would have to be
specified with finite precision. This implies that the input set~$S$
(following the notation in Section~\ref{part:Approximate-remote-state})
is necessarily countable.  This point has not been addressed in previous
works.

To meaningfully consider the preparation states drawn from an
uncountable set, we may instead consider \emph{approximations\/} 
drawn from a suitable countable set. For example, instead of the
set~$\sD(\cH)$ of all quantum states over a~$d$-dimensional space~$\cH$,
we may instead study the countably dense set of states whose
matrix representations only have complex entries with rational real and
imaginary parts. Such states have unique finite-length representations.
(Similar approximation is also implicit in the case of RSP of a finite
set of states, 
when the corresponding matrices involve irrational numbers.)

Another approach, perhaps only of theoretical interest, would be to
allow the local operations in an LOCC protocol to be defined on a 
suitable generalization of the Real RAM model due to Blum, Shub, and 
Smale~\cite{BSS89}. We do not attempt to define such a model of
computation here. For our purposes, it would suffice to assume a model which enables the implementation of quantum operations such as
unitary operations controlled by the registers holding real numbers in
finite time.

We assume that we take one of the abovementioned approaches in the
analysis in this section. The underlying idea, that of approximating
states from an infinite set with those from a \emph{net\/}, 
probably applies in other reasonable approaches as well. 

As before, we restrict ourselves to states over a finite dimensional 
Hilbert space~$\cH$.

\begin{definition}
Let~$\nu \in (0,1]$ and~$D \subseteq \sD(\cH)$ be any set of quantum
states. A~\emph{$\nu$-net\/}~$N$ in~$D$ is a subset of~$D$ such that for
any state~$\rho \in D$, there is a state~$\sigma \in N$ such
that~$\rP(\rho,\sigma) < \nu$.
\end{definition}

We argue that every subset of finite-dimensional states admits a finite
net.

\begin{proposition}
For every~$\nu \in (0,1]$, and every set~$D \subseteq \sD(\cH)$ of
quantum states, there is a finite~$\nu$-net in~$D$.
\end{proposition}
\begin{proof}
Since~$\sD(\cH)$ is compact, it has a finite cover~$(B_i)$ consisting of
open balls of radius~$\nu/2$. This is also 
a cover for any subset~$D$ of quantum states. Let~$N$ be a subset of~$D$
constructed by taking one point from~$B_i \intersect D$, whenever this
intersection is non-empty. We claim that this is a finite~$\nu$-net
in~$D$. 

Consider a state~$\rho \in D$. Since~$(B_i)$ is a cover for the set of
all quantum states, $\rho \in B_j$ for some~$j$. By construction, there
is a state~$\sigma \in N$ from~$B_j \intersect D$. Since~$\rho,\sigma$
both belong to the same ball~$B_j$ of radius~$\nu/2$, we 
have~$\rP(\rho,\sigma) < \nu$.  So~$N$ is a~$\nu$-net in~$D$.
\end{proof}

Suppose~$S$ is an infinite set, and~$Q:S\rightarrow \sD(\cH)$ is a 
one-to-one function mapping each element of~$S$ to a quantum state.
(We view an element~$x \in  S$ as a description, i.e., unique encoding,
of the quantum state~$Q(x)$.) Define~$R \coloneqq Q(S)$ as
the image of~$S$ under~$Q$; this is the set of quantum states under
consideration. We fix an approximation parameter~$\nu > 0$ of
our choice, and a finite~$\nu$-net~$N$ in~$R$, and let~$T
\coloneqq Q^{-1}(N)$ be the set of inputs corresponding to~$N$.
We bound the communication
required for remote state preparation of states from~$R$ with that for states
from~$N$. We may then appeal to Theorem~\ref{thm-main theorem} to infer
bounds on $\RSP(S,Q)$.

\paragraph{Worst-case error.}
We first consider the simpler case, that of worst-case error~$\epsilon > 0$.
Any protocol for $\RSP(S,Q)$ with worst-case error~$\epsilon$ is also a 
protocol for $\RSP(T,Q)$ as~$T$ is a subset of~$S$. So we have
\[
\sQ^*(\RSP(T,Q),\epsilon) \quad \leq \quad \sQ^*(\RSP(S,Q),\epsilon)
\enspace.
\]

Now suppose~$\Pi$ is a protocol for~$\RSP(T,Q)$ with communication 
cost~$c$ and worst case error~$\epsilon$. We design a protocol~$\Pi'$
for~$\RSP(S,Q)$ as follows. Given an~$x\in S$, Alice chooses~$y\in T$ such 
that~$\rP(Q(x),Q(y))\leq \nu$, and prepares an approximation of~$Q(y)$
on Bob's side using protocol~$\Pi$. Suppose Bob's output is~$\sigma_y$.
Then
\[
\rP(Q(x),\sigma_y) \quad \leq \quad \rP(Q(x),Q(y))+\rP(Q(y),\sigma_y)
\quad \leq \quad \nu + \epsilon \enspace.
\]
So~$\Pi'$ is a protocol for $\RSP(S,Q)$ with communication cost~$c$, 
and worst case error~$\epsilon+\nu$.
Therefore, 
\[
\sQ^*(\RSP(S,Q), \epsilon + \nu) \quad \leq \quad \sQ^*(\RSP(T,Q),\epsilon)
\enspace.
\]
Putting the two together, for~$\nu, \epsilon$ such
that~$0 < \nu < \epsilon$, we get
\[
\sQ^*(\RSP(T,Q),\epsilon) 
    \quad \leq \quad \sQ^*(\RSP(S,Q), \epsilon) 
    \quad \leq \quad \sQ^*(\RSP(T,Q),\epsilon - \nu) \enspace.
\]

\paragraph{Average-case error.}
Next we consider approximate RSP with average error at most~$\epsilon
\in (0,1]$ with respect to a probability measure~$\mu$ on the set of
states~$R$. For simplicity, we only consider the case when
the open sets in~$R$ generated by the metric~$\rP$ are measurable.
Since~$Q$ is injective, we may equivalently consider~$\mu$ as a 
probability measure on~$S$.

Let~$(\rho_i)$ be an enumeration of the states in~$N$,
and~$(B_i)$ be open balls of radius~$\nu$ centred at~$\rho_i$ with
respect to the metric~$\rP$. Since~$N$ is a~$\nu$-net in~$R$, we have~$R 
\subseteq \union_i B_i$.
Define the function~$f : R \rightarrow N$ as~$f(\sigma) 
\coloneqq \rho_i$ for all states~$\sigma \in (B_i
\intersect R) \setminus (\union_{j < i} B_j)$. The function~$f$
maps each quantum state~$\rho \in R$ to a quantum state in~$N$ such
that~$\rP(\rho,f(\rho)) < \nu$. Moreover, it is measurable.

The function~$f$ induces a probability distribution~$p$ on~$N$
in the natural way:
\[
p_{\rho_i} \quad \coloneqq \quad \mu(f^{-1}(\rho_i))
\]
for~$\rho_i \in N$. We may view the distribution~$p$ as being over the
corresponding set~$T$ of inputs: for~$y \in T$ such that~$Q(y) = \rho_i$, we
define~$p_y \coloneqq p_{\rho_i}$. 

We relate protocols for $\RSP(S,Q)$ with average error~$\epsilon$ 
with respect to~$\mu$ to protocols for~$\RSP(T,Q)$ with average error
``close'' to~$\epsilon$ with respect to~$p$.

\begin{lemma}
	Suppose~$\Pi$ is a protocol for~$\RSP(S,Q)$ with communication 
cost~$c$ and average error~$\epsilon$ with respect to~$\mu$. Then there is
a protocol~$\Pi'$ for~$\RSP(T,Q)$ with communication cost~$c$ and average 
error at most~$\nu+\epsilon$ with respect to~$p$.
\end{lemma}

\begin{proof}
For~$y \in T$,  define~$R_y \coloneqq f^{-1}(Q(y))$, the set of 
states in~$R$ that are mapped to~$Q(y) \in N$. Define~$S_y \coloneqq
Q^{-1}(R_y)$, the set of inputs corresponding to the states in~$R_y$.
Note that~$(R_y)$ is a partition of~$R$ and~$(S_y)$ of~$S$. Since~$f$ is
measurable, $R_y$ is a measurable set. When~$R_y$ has non-zero measure, 
we define a probability measure~$\mu_y$ on~$R_y$ as~$\mu_y(W) \coloneqq
\mu(W)/ \mu(R_y)$ for all measurable sets~$W \subseteq R_y$. We also
view~$\mu_y$ as a probability measure on~$S_y$.

We now construct the protocol~$\Pi'$ for~$\RSP(T,Q)$ as follows.
Given~$y \in T$, Alice selects an input~$x \in S_y$ randomly with 
respect to the probability measure~$\mu_y$ and runs the protocol~$\Pi$
on this input.

The communication in~$\Pi'$ is also~$c$. Suppose~$\sigma_x$ is the output 
of the protocol~$\Pi$ when the input is~$x$. Then the average error of 
the protocol~$\Pi'$ is
\begin{align*}
\sum_{y \in T} p_y \int_{x \in S_y} \rP(Q(y),\sigma_{x})\ \rd\mu_y(x) \quad 
    & = \quad \sum_{y \in T} ~ \int_{x \in S_y} 
        \rP(Q(y),\sigma_{x})\ \rd\mu(x) \\
    & \leq \quad \sum_{y \in T} ~  \int_{x \in S_y} 
        \rP(Q(y),Q(x)) \ \rd\mu(x) \\
    & \qquad \mbox{} + \sum_{y \in T} ~ \int_{x \in S_y} 
        \rP(Q(x),\sigma_{x})\ \rd\mu(x) \\
    & \leq \quad \nu+\epsilon \enspace,
\end{align*}
as claimed.
\end{proof}

Conversely, we can also derive a protocol for $\RSP(S,Q)$ from one
for~$\RSP(T,Q)$.

\begin{lemma}
	Suppose~$\Pi$ is a protocol for~$\RSP(T,Q)$ with 
communication cost~$c$ and average error~$\epsilon$ with respect to
the distribution~$p$. There exists a protocol~$\Pi'$ for~$\RSP(S,Q)$ with 
communication cost~$c$ and average error at most~$\epsilon
+ \nu$ with respect to~$\mu$.
\end{lemma}

\begin{proof}
In the protocol~$\Pi'$, given input~$x\in S$, Alice runs the 
protocol~$\Pi$ on input~$y$ defined as~$y \coloneqq Q^{-1}(f(Q(x)))$.
This is the input corresponding to the state in the~$\nu$-net to
which~$f$ maps~$Q(x)$. The communication cost of~$\Pi'$ is also~$c$.

Suppose the output of~$\Pi'$ on input~$x$ is~$\sigma_x$.
Note that~$f$ maps all states~$Q(x)$ for~$x \in S_y$ to the same 
value~$Q(y)$, and therefore the outputs~$\sigma_x$ for all 
inputs~$x \in S_y$ are equal to~$\sigma_y$.

The average error of the protocol with respect to~$\mu$ is
\begin{align*}
\int_{x\in S} \rP(Q(x),\sigma_x )\ \rd\mu(x) \quad
    & \leq \quad \int_{x\in S} \rP(Q(x), f(Q(x)))\ \rd\mu(x)
        + \int_{x\in S} \rP(f(Q(x)),\sigma_{x})\ \rd\mu(x) \\
    & \leq \quad \nu + \sum_{y \in T} ~ \int_{x \in S_y} 
        \rP(f(Q(x)),\sigma_{x}) \ \rd\mu(x) \\
    & = \quad \nu + \sum_{y \in T} ~ \int_{x \in S_y} \rP(Q(y),\sigma_{y})
        \ \rd\mu(x) \\
    & = \quad \nu + \sum_{y \in T} p_y \, \rP(Q(y),\sigma_{y}) \\
    & \leq \quad \nu+\epsilon \enspace,
\end{align*}
where in the third step we have used the abovementioned
property that~$f$ is constant on~$S_y$.
\end{proof}

For~$\nu,\epsilon$ such that~$\nu<\epsilon$, the above two lemmata imply that
\[
\sQ^*_{p}(\RSP(T,Q),\epsilon+\nu) \quad \leq \quad \sQ^*_{\mu}(\RSP(S,Q),\epsilon)\quad \leq \quad \sQ^*_{p}(\RSP(T,Q),\epsilon-\nu) \enspace.
\]

\end{document}